\documentclass[conference]{IEEEtran}

\usepackage{hyperref}
\usepackage{macros}
\usepackage{stringdiagrams}
\usepackage{notesfigures}
\usepackage{listings}

\usepackage{txfonts}

\lstnewenvironment{code}
    {\lstset{}%
      \csname lst@SetFirstLabel\endcsname}
    {\csname lst@SaveFirstLabel\endcsname}
\newtheorem{theorem}{Theorem}[section] %
\newtheorem{proposition}[theorem]{Proposition}

\newtheorem{lemma}[theorem]{Lemma}
\newtheorem{corollary}[theorem]{Corollary}
\theoremstyle{definition}
\newtheorem{definition}[theorem]{Definition}

\theoremstyle{remark}
\newtheorem{remark}[theorem]{Remark}

\newtheorem{example}[theorem]{Example}
\allowdisplaybreaks{}

\title{Evidential Decision Theory\texorpdfstring{\\}{}via Partial Markov Categories}

\author{Elena Di Lavore and Mario Román\\ \textit{Tallinn University of Technology, Tallinn, Estonia}}

\IEEEoverridecommandlockouts
\IEEEpubid{\makebox[\columnwidth]{979-8-3503-3587-3/23/\$31.00~
\copyright2023 IEEE \hfill} \hspace{\columnsep}\makebox[\columnwidth]{ }}
\begin{document}
\maketitle

\begin{abstract}
  We introduce partial Markov categories.
  In the same way that Markov categories encode stochastic processes, partial Markov categories encode stochastic %
  processes with constraints, observations and updates.
  In particular, we prove a synthetic Bayes theorem and we apply it to define a syntactic partial theory of observations on any Markov category whose normalisations can be computed in the original Markov category.
  Finally, we formalise Evidential Decision Theory in terms of partial Markov categories, and provide examples.
\end{abstract}

\let\thefootnote\relax\footnotetext{\emph{Funding.} Elena Di Lavore and Mario Rom\'an were supported by the ESF funded Estonian IT Academy research measure (project 2014-2020.4.05.19-0001) and the Estonian Research Council grant PRG1210.}

\section{Introduction}\label{sec:introduction}
Evidential Decision Theory is a branch of decision theory that focuses on \emph{observational evidence}: given a decision problem, Evidential Decision Theory prescribes the action that we \emph{observe to have done} in the best possible outcome~\cite{ahmed_2014}.
This contrasts with Causal Decision Theory, which prescribes the action that \emph{causes} the best possible outcome~\cite{gibbard_harper_1978}.
In Evidential Decision Theory, no direct causal connection is required for the action to affect the outcome: it suffices that the observation of the action alters the conditional probability of the outcome via Bayesian update~\cite{yudkowsky_soares_2017}.
Characterizing, comparing and formalizing decision theories, such as Evidential Decision Theory, remains an open problem in artificial intelligence research~\cite{horvitz88:decision,everitt15:sequential}.

Newcomb's paradox~\cite{nozick_1969} is a famous decision problem that sets apart Evidential and Causal Decision Theory.
An agent is in front of two boxes: a transparent box filled with \(1 \matheuro\) and an opaque box.
The agent is given the choice between taking both boxes (\emph{two-boxing}) and taking just the opaque box (\emph{one-boxing}).
However, the opaque box is controlled by a very accurate predictor.
The predictor placed \(1000 \matheuro\) in the box if it predicted that the agent would one-box and left it empty otherwise.
The agent knows this.
Which action should the agent choose?
\begin{table}[h!]
  \centering
  \caption{Utilities for Newcomb's problem.}\label{fig:newcomb-utilities}
  \begin{tabular}{|c | c c |}
    \hline
    & predictor = \emph{one-box} & predictor = \emph{two-box}\\
    \hline
    agent = \emph{one-box} & \(1000 \matheuro\) & \(0 \matheuro\)\\
    agent = \emph{two-box} & \(1001 \matheuro\) & \(1 \matheuro\)\\
    \hline
  \end{tabular}
\end{table}

Evidential Decision Theory asks: \emph{``Which action is evidence for the best possible outcome?''}
In the case of Newcomb's paradox, %
because the predictor is accurate, actions and predictions are correlated.
\emph{Two-boxing} correlates with the opaque box being empty; \emph{one-boxing} correlates with this box being full.
Therefore, Evidential Decision Theory prescribes one-boxing: an evidential decision theorist gets \(1000\matheuro\) almost all the times.
On the other hand, Causal Decision Theory prescribes the action that \emph{causes} the best possible outcome.
For Newcomb's paradox, %
assuming any of the two predictions, two-boxing has a strictly greater utility. 
Therefore, the %
causal decision theorist is bound to two-boxing and getting \(1\matheuro\) almost all the times.

Formalising decision problems is subtle: different decision theories disagree in many well-studied scenarios~\cite{gibbard_harper_1978}, and solutions are sensitive to slight modifications in the problem statement.
In order to clarify these disagreements, we need both an intuitive mathematical syntax to model decision problems, and a formal algorithmic procedure to solve them according to the prescriptions of the theory.
In this paper, we answer the following question:
\slogan{What is a minimalistic mathematical framework\\ that can formulate and solve decision problems\\ in Evidential Decision Theory?}

A good calculus for modelling decision problems needs to \emph{(i)} express probabilistic processes, to model a stochastic environment in which the agent needs to act, \emph{(ii)} express constraints, to restrict the model to satisfy the specifications of the decision problem, and \emph{(iii)} explicitly capture the implicit assumptions of decision theory.
We introduce partial Markov categories as both a syntax for modelling decision problems and a calculus for solving them.

Markov categories~\cite{fritz_2020} are a syntax for probabilistic processes, where it is natural to express \emph{conditioning}, \emph{independence} and \emph{Bayesian networks}~\cite{fong_2013,cho_2019,fritz_2020,jacobs2021causal,jacobs2020logical}.
For example, the predictor, the agent and the utility function in Newcomb's paradox can be expressed as morphisms in a suitable Markov category.
However, this language does not allow the encoding of constraints, e.g., in the context of Newcomb's paradox, we cannot impose \emph{a posteriori} that the predictor's prediction very likely coincides with the agent's action.

This restriction is imposed by the structure of Markov categories that allows resources and processes to be \emph{discarded} but only resources to be \emph{copied}: throwing a coin twice is not the same as throwing it once and copying the result.
String diagrams are the internal language for morphisms in monoidal categories and for Markov categories in particular.
\Cref{fig:copy-discard-markov} shows the string diagrams for discarding and copying morphisms in a Markov category.
\begin{figure}[h!]
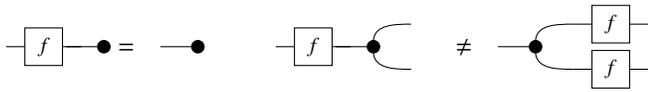

  \[\copydiscardmarkovFig{}\]
  \caption{Stochastic processes can be discarded but not copied.}\label{fig:copy-discard-markov}
\end{figure}

Most Markov categories commonly used to encode stochastic processes exhibit the additional property of having \emph{conditionals} (\Cref{fig:conditionals}).
Conditionals ensure that every joint distribution can be split into a marginal and a conditional distribution.
While not necessary to encode stochastic processes, this property is essential for \emph{reasoning} about them. In this manuscript, we always consider Markov categories to have conditionals.\footnote{In standard terminology, Markov categories are ``copy and natural-discard categories with conditionals'' and partial Markov categories are ``copy-discard categories with conditionals''.}
\begin{figure}[h!]
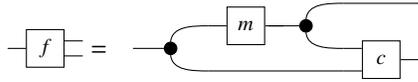

  \[\generatorFig{morphismOneTwo=\(f\)} = \conditionalsFig{m}{c}{}{}{}\]
  \caption{Conditionals require that a stochastic process \(f\) be split into a marginal \(m\) and a conditional \(c\).}\label{fig:conditionals}
\end{figure}

On the other hand, \emph{cartesian restriction categories}~\cite{cockett2007restriction, curien_obtulowicz_1989}, and discrete cartesian restriction categories in particular~\cite{cockett2012range2, di_liberti_nester_2021}, are a calculus of partial processes with \emph{constraints}.
Cartesian restriction categories allow the copying and discarding of resources, as in Markov categories.
However, processes are now allowed to be copied but not to be discarded (\Cref{fig:copy-discard-partial}).
\begin{figure}[h!]
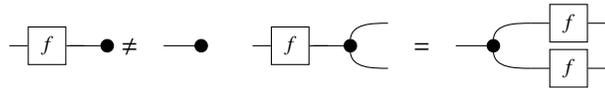

  \[\copydiscardpartialFig{}\]
  \caption{Partial processes can be copied but not discarded.}\label{fig:copy-discard-partial}
\end{figure}

\emph{Discrete} cartesian restriction categories additionally possess an \emph{equality constraint} (\textmonoidpic[copycolor]).
Its axioms say that copying resources and then checking that they are equal should be the same as the identity process (\Cref{fig:comparator-axioms}).
\begin{figure}[h!]
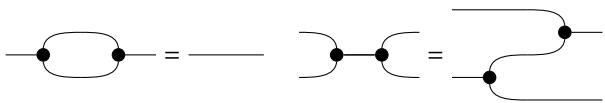

  \comparatoraxiomsFig{}
  \caption{Equality constraint structure and its axioms in a discrete cartesian restriction category.}\label{fig:comparator-axioms}
\end{figure}

We introduce \emph{partial Markov categories} and \emph{discrete partial Markov categories}, which extend Markov categories and (discrete) cartesian restriction categories to encode both \emph{probabilistic reasoning} and \emph{constraints}.
Discrete partial Markov categories allow the modelling of decision problems like Newcomb's paradox: they can express the constraint about the prediction matching the action of the agent (\Cref{fig:newcomb-diagram}).
\begin{figure}[!ht]
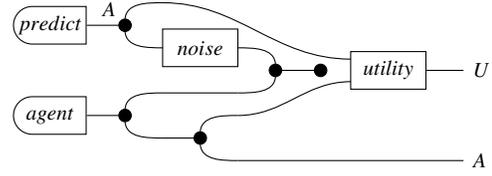

  \[\newcombFig{}\]
  \caption{Model of Newcomb's paradox in a discrete partial Markov category.}\label{fig:newcomb-diagram}
\end{figure}
With the same technique, discrete partial Markov categories can solve decision problems according to Evidential Decision Theory:
computing the utility constrained to a certain action allows us to answer %
\emph{``which action we observe to have done in the best possible outcome?''}

That is, observing the value of a variable means constraining the variable to take a certain value.
In the case of Newcomb's paradox, the solution can be computed by simplifying the diagram in \Cref{fig:newcomb-solution} according to the axioms of discrete partial Markov categories, as we do in \Cref{sec:solvingdecision}.
\begin{figure}[h!]
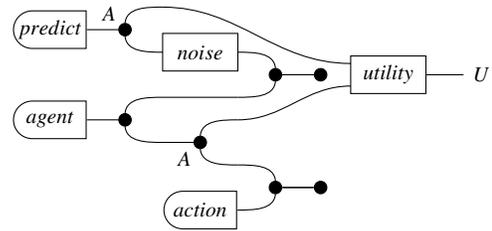

  \[\newcombsolutionFig{}\]
  \caption{Newcomb's problem in a discrete partial Markov category.}\label{fig:newcomb-solution}
\end{figure}

We claim that the algebra of partial Markov categories is a good theoretical framework for Evidential Decision Theory:
it provides both a convenient syntax in terms of string diagrams and a calculus that translates these diagrams to the computations that solve the given decision problem. %

\subsection{Contributions}
Our main conceptual contribution is the algebra of partial Markov categories (\Cref{def:partial-markov-cat}) and discrete partial Markov categories (\Cref{def:discrete-partial-markov-cat}). 
Our main result is the construction of a partial Markov category on top of any Markov category, such that deterministic observations are computed in terms conditionals of the base Markov category (\Cref{prop:normalisation-of-programs,prop:conditionals-of-programs}).

We introduce partial and discrete partial Markov categories in \Cref{sec:partial-markov}, where we show (\Cref{th:bayes}) a synthetic version of Bayes' theorem that holds in any discrete partial Markov category.
We apply this framework to synthetically compare Pearl's and Jeffrey's update rules (\Cref{def:pearl-update,def:jeffrey-update,prop:pearlandjeffrey}), to model decision problems like Newcomb's paradox, and to solve them according to Evidential Decision Theory (\Cref{sec:evidential-dec-theory}).
We show that the Kleisli category of the \emph{Maybe} monad over a Markov category is a partial Markov category (\Cref{th:maybe-subdistributions}).

Finally, in \Cref{sec:observations-conditionals}, we provide a calculus for reasoning with exact conditioning (\Cref{def:exact-conditioning-cat}).
The construction from \Cref{prop:normalisation-of-programs,prop:conditionals-of-programs} allows us to express deterministic observations even in non-discrete cases like that of continuous probabilistic processes.

\subsection{Related work}
\paragraph*{Markov categories}
The categorical approach to probability theory based on Markov categories~\cite{fritz_2020} has led to the abstraction of various results from probability theory~\cite{fritz2020infinite,fritz2021probability,fritz2019probability,fitz2021deFinetti}.
Markov categories have been further applied for formalising Bayes networks and other kinds of probabilistic reasoning in categorical terms~\cite{fong_2013, jacobs2020logical, jacobs2021causal}.
The breadth of results and applications of Markov categories suggest that there can be an equally rich landscape for their partial counterpart.

\paragraph*{Categories of partial maps}
Partiality has long been studied in Computer Science and even categorical approaches to it date back to the works of Carboni~\cite{carboni1987bicategories}, Di Paola and Heller~\cite{di1987dominical}, Robinson and Rosolini~\cite{robinson1988categories}, and Curien and Obtu{\l}owicz~\cite{curien_obtulowicz_1989}.
However, our categorical structures are more related to more recent work on restriction categories~\cite{cockett02,cockett2003restriction}, and, in particular, cartesian restriction categories~\cite{cockett2007restriction} and discrete cartesian restriction categories~\cite{cockett2012range2,di_liberti_nester_2021}.

\paragraph*{Copy-discard categories}
Partial categories for probabilistic processes have been considered previously~\cite{panangaden1999, jacobs2018probability, cho_2019} but no comprehensive presentation was given.
Bayesian inversion for compact closed copy-discard categories has been studied by Coecke and Spekkens~\cite{coeckespekkens2012picturing}.
The relationship between that definition and \Cref{def:bayes-inversion} might involve normalisation.
Copy-discard categories have been applied to graph rewriting~\cite{corradini_gadducci_1999}, where they are called GS-monoidal.

\paragraph*{Categorical semantics of probabilistic programming}
There exists a vast literature on categorical semantics for probabilistic programming languages (for some related, see e.g.~\cite{hasegawa97, stay2013bicategorical, staton_et_al_2016, heunen_kammar_staton_yang_2017, ehrhard2017measurable, dahlqvist19semantics, vakar2019domain}).
However, while the internal language of Markov categories has been studied, the notion of partial Markov category and its diagrammatic syntax have remained unexplored.
Stein~\cite{stein2021compositional, stein_thesis_2021} has recently presented the $\mathbf{Cond}$ construction for exact conditioning, which could be related to our construction of a partial Markov category of \constrainedProcesses{} in \Cref{def:exact-conditioning-cat} via normalisation.%

\paragraph*{Evidential Decision Theory}
Existing formalisations of decision theories exist mainly in philosophical terms~\cite{lewis1981causal, joyce1999foundations, greaves2006justifying, ahmed_2014, greaves2013epistemic, yudkowsky_soares_2017}.
We provide a categorical formalisation of Evidential Decision Theory.
 \section{Preliminaries}\label{sec:preliminaries}
Symmetric monoidal categories are an algebra of processes that compose sequentially and in parallel.
They possess a convenient sound and complete syntax in terms of string diagrams~\cite{joyal91}.
In particular, copy-discard categories and the more specialized Markov categories allow us to reason about probabilistic processes.
In this section, we introduce the categorical approach to probability theory.

\subsection{Copy-discard categories}
Copy-discard categories is the name we give to any category where each object has a uniform comonoid structure (\Cref{diagram-copy-discard}).
The comultiplication is what we call the ``copy'' and the counit is what we call the ``discard'': they have the same operations as cartesian monoidal categories, but neither is assumed to be natural.
Copy-discard categories have been called GS-monoidal categories when applied to graph rewriting~\cite{corradini_gadducci_1999} and CD categories when applied to non-normalised probabilistic processes~\cite{cho_2019}. See~\cite[Remark 2.2]{fritz_liang_2022} for a history of the term.

\begin{figure}[!ht]
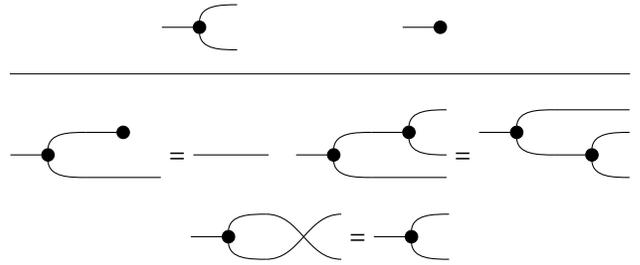

  \cocommutativecomonoidAxiomsFig{}
  \caption{Generators and axioms of a cocommutative comonoid.}\label{diagram-copy-discard}
\end{figure}
\begin{definition}
  A \emph{\defining{linkcopy}{copy}-\defining{linkdiscard}{discard} category} is a symmetric monoidal category \(\cat{C}\) where every object \(X \in \cat{C}\) has a cocommutative comonoid structure \((X, \cp_X, \discard_X)\) (\Cref{diagram-copy-discard}) and this structure is uniform: \(\cp_{X \tensor Y} = (\cp_X \tensor \cp_Y) \comp (\id{X} \tensor \swap[X,Y] \tensor \id{Y})\), \(\cp_{\monoidalunit} = \id{}\), \(\discard_{X \tensor Y} = \discard_X \tensor \discard_Y\), and \(\discard_{\monoidalunit} = \id{}\)\footnote{We omit associators and unitors to avoid clutter. By the coherence theorem of monoidal categories~\cite{macLane1971}, associators and unitors can be recovered uniquely.}.
\end{definition}

Copying and discarding are not required to be natural: only some of the morphisms will be copyable or discardable.
We call these \emph{deterministic} and \emph{total}, respectively.

\begin{definition}\label{def:total_deterministic}
  A morphism \(f \colon X \to Y\) in a copy-discard category is called \emph{deterministic} if \(f \comp \cp_{Y} = \cp_{X} \comp (f \tensor f)\) (\Cref{diagram-total-deterministic}, left); and \emph{total} if \(f \comp \discard_{Y} = \discard_{X}\) (\Cref{diagram-total-deterministic}, right).
  \begin{figure}[h!]
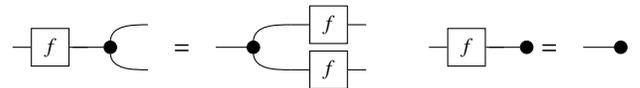

    \begin{align*}
      \copynaturalFig{} & \discardnaturalFig{}
    \end{align*}
    \caption{Deterministic and total morphisms.}\label{diagram-total-deterministic}
\end{figure}
\end{definition}

\subsection{Markov categories}
Probability theory requires more structure than the one given by copy-discard monoidal categories. Explicitly, it is usually assumed that a category that encodes a theory of probability will have a notion of \emph{conditional}~\cite{cho_2019,fritz_2020}.

Markov categories~\cite{fritz_2020} have been defined as copy-discard categories where the counit is moreover natural.
The Markov categories better suited for probability theory are those that have \emph{conditionals}.
We decide to call \emph{Markov categories} only to those with conditionals. The purpose of this slight change of convention is to make the parallel with cartesian categories more explicit: Markov categories are cartesian categories with a weaker splitting.

\begin{definition}
  A copy-discard category \(\cat{C}\) has \emph{conditionals} if, for every morphism \(f \colon X \to Y_{1} \tensor Y_{2}\), there are \(c \colon Y_{1} \tensor X \to Y_{2}\) and \(m \colon X \to Y_{1}\) such that
  \(f = \cp \dcomp (m \tensor \id{}) \dcomp (\cp \tensor \id{}) \dcomp (\id{} \tensor c)\), i.e.~they satisfy the equation in \Cref{fig:conditionals}.
\end{definition}

In this situation, $c \colon Y_1 \tensor X \to Y_2$ is a \emph{conditional} of $f$ with respect to $Y_1$ and $m \colon X \to Y_1$ is a \emph{marginal} of $f$ on $Y_1$.
Note that, in general, conditionals and marginals are not unique.
\Cref{prop:unique-conditionals-collapse} investigates the consequences of unique conditionals.

\begin{definition}
  A \emph{Markov category} is a copy-discard category with conditionals where all morphisms are total (\Cref{diagram-markov-split}).
\end{definition}
\begin{figure}[h!]
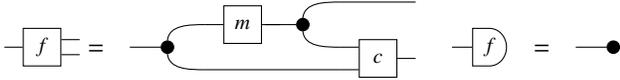

  \begin{align*}
    \conditionalsdefFig{}&\discardfactorisationFig{}
  \end{align*}
  \caption{Markov category with conditionals.}\label{diagram-markov-split}
\end{figure}

\begin{remark}[Notation]\label{rem:conditional-composition}
  For \(m \colon X \to A\) and \(c \colon A \tensor X \to B\), we indicate with \(m \condcomp c \colon X \to A \tensor B\) the \defining{linkcondcomp}{composition} of a marginal with its conditional: \(m \condcomp c \defn \cp_{X} \dcomp (m \tensor \id{X}) \dcomp (\cp_{A} \tensor \id{X}) \dcomp (\id{A} \tensor c)\), as shown in \Cref{diagram-markov-split}, left.
\end{remark}

\subsection{Some Markov categories} 

The canonical example of Markov category is the monoidal Kleisli category of the finitary distribution monad~\cite{fritz_2020}.

\begin{definition}\label{def:finitary-distributions}
  A \emph{finite support distribution} \(\sigma\) on a set \(X\) is a function \(\sigma \colon X \to [0,1]\) such that the set \(\{x \in X \st \sigma(x) \neq 0\}\) is finite and \(\sum_{x \in X} \sigma(x) = 1\).
  We indicate the set of finite support distributions on a set \(X\) as \(\distr(X)\).
  The mapping \(\distr\) can be extended to a functor \(\distr \colon \Set \to \Set\) and to a monad, called the \defining{linkdistr}{\emph{finitary distribution monad}}.
  For a function \(f \colon X \to Y\), the function \(\distr(f)\) associates to a distribution \(\sigma \in \distr(X)\) a distribution \(\tau \in \distr(Y)\) defined by \(\tau(y) \defn \sum_{x \st f(x) = y} \sigma(x)\).
  The monad multiplication \(\mu_{X} \colon \distr(\distr(X)) \to \distr(X)\) associates to \(f \in \distr(\distr(X))\) the distribution \(g \in \distr(X)\) given by \(g(x) \defn \sum_{\sigma \in \distr(X)}\sigma(x) \cdot f(\sigma)\).
  The monad unit associates to each \(x \in X\) the Dirac delta \(\dirac{x}\).
  The finitary distribution monad is monoidal with structural transformation \(\laxator_{X,Y} \colon \distr(X) \tensor \distr(Y) \to \distr(X \tensor Y)\) given by \(\laxator_{X,Y}(\sigma,\tau) (x,y) \defn \sigma(x) \cdot \tau(y)\).
\end{definition}

The Kleisli category of \(\distr\) is a copy-discard category with conditionals. %
The copy-discard structure is lifted from \(\Set\) by post-composing with the unit of the finitary distribution monad \(\distr\).
A morphism \(f \colon X \to Y\) in \(\kleisli{\distr}\) represents a stochastic channel: we interpret the value of \(f(x)\) in \(y \in Y\) as the probability of \(y\) given \(x\) according to the channel \(f\), and we indicate it as \(f(y \given x)\).

\begin{remark}
  The definition of composition in \(\kleisli{\distr}\) and its copy-discard structure allow an intuitive reading of a string diagram in terms of its components.
  The value of a morphism is obtained by multiplying the values of all the components and summing on the wires that are not inputs or outputs.
  For example, the morphism
  \[\morphismcopydiscardExampleFig{}\] 
  evaluates to the formula 
  $$f (z_{1}, z_{2} \given x) = \sum_{y \in Y} g(y \given x) \cdot h_{1}(z_{1} \given y) \cdot h_{2}(z_{2} \given y).$$
\end{remark}

Measurable maps between measurable spaces %
form a copy-discard category.
However, this category does not have conditionals. Instead, we consider its subcategory on standard Borel spaces, which has conditionals and is a Markov category~\cite{fritz_2020}.

\begin{definition}
  The \defining{linkborelstoch}{category} \(\BorelStoch\) has standard Borel spaces \((X,\sigma_{X})\), where \(X\) is a set and \(\sigma_{X}\) is a \(\sigma\)-algebra on \(X\), as objects.
  A morphism \(f \colon (X,\sigma_{X}) \to (Y,\sigma_{Y})\) is a function \(f \colon \sigma_{Y} \times X \to [0,1]\) such that, for each \(T \in \sigma_{Y}\), \(f(T \given -) \colon X \to [0,1]\) is a measurable function, and, for each \(x \in X\), \(f(- \given x) \colon \sigma_{Y} \to [0,1]\) is a probability measure.
  The composition of \(f \colon (X, \sigma_{X}) \to (Y, \sigma_{Y})\) and \(g \colon (Y, \sigma_{Y}) \to (Z, \sigma_{Z})\) and the identity \(\id{X} \colon (X, \sigma_{X}) \to (X, \sigma_{X})\) are given by
  \begin{align*}
    (f \dcomp g) (T \given x) &\defn \int_{y \in Y} g(T \given y) \cdot f(dy \given x),\\
    \id{X} (T \given x) &\defn \begin{cases} 1 & \text{ if } x \in T, \\ 0 & \text{ otherwise. }\end{cases}
  \end{align*}
  The category \(\BorelStoch\) is monoidal.
  The monoidal product is defined on objects by \((X, \sigma_{X}) \tensor (Y, \sigma_{Y}) \defn (X \times Y, \sigma_{X} \times \sigma_{Y})\) and on morphisms by
  \[(f \tensor f')(T \given x,x') \defn \int_{(y,y') \in T} f(dy \given x) \cdot f'(dy' \given x').\]
  The monoidal unit is the one-element set.
\end{definition}

\begin{definition}
  The \defining{linkgiry}{Giry functor} \(\Giry \colon \Meas \to \Meas\) assigns to a set \(X\) the set of probability measures on it.%
  On morphisms, \(f \colon (X, \sigma_{X}) \to (Y, \sigma_{Y})\), it is defined by
  \(\Giry(f) (T \given \sigma) \defn \sigma(f^{-1}(T))\).
  This functor is a monad~\cite{giry82:categorical}: its unit \(\monadunit_{(X,\sigma_{X})} \colon x \mapsto \dirac{x}\) that associates to each \(x \in X\) the Dirac distribution at \(x\); its multiplication is \(\monadmultiplication_{(X,\sigma_{X})} (T \given p) \defn \int_{\tau \in \Giry(X,\sigma_{X})} \tau(T) \cdot p(d\tau)\).
  The Giry monad is monoidal with a structural transformation analogous to that of the finitary distribution monad: \(\laxator_{X,Y}(\sigma,\tau) (S,T) \defn \sigma(S) \cdot \tau(T)\).
\end{definition}

\begin{remark}
  The category \(\BorelStoch\) can also be seen as the Kleisli category of the Giry monad \(\Giry \colon \Meas \to \Meas\) restricted to standard Borel spaces~\cite{giry82:categorical,fritz_2020}.
\end{remark}

\subsection{Subdistributions}\label{ex:non-example-markov}

A \emph{subdistribution} \(\sigma\) over \(X\) is a distribution whose total probability is allowed to be less than \(1\)~\cite{jacobs2018probability,cho_2019}.
In other words, it is a distribution over \(\maybe[X]\).
This means that a morphism \(f \colon X \to Y\) in \(\kleisli{\subdistr}\) represents a stochastic channel that has some probability of failure.

The symmetric monoidal Kleisli category of the finitary subdistribution monoidal monad, \(\subdistr\), is the main example for (discrete) partial Markov categories.
It is the semantic universe where we compute the solutions to the decision problems in \Cref{sec:evidential-dec-theory}.

\begin{definition}
  A \emph{finite support subdistribution} on a set \(X\) is a function \(\sigma \colon X \to [0,1]\) such that the set \(\{x \in X \st \sigma(x) \neq 0\}\) is finite and \(\sum_{x \in X} \sigma(x) \leq 1\).
  We indicate the set of subdistributions on a set \(X\) as \(\subdistr(X)\).
  The mapping \(\subdistr\) can be extended to a functor \(\subdistr \colon \Set \to \Set\) and to a monad, called the \defining{linksubdistr}{\emph{finitary subdistribution monad}}.
  For a function \(f \colon X \to Y\), \(\subdistr(f)\) is defined by
  \[\subdistr(f)(y \given \sigma) \defn \sum_{f(x)=y} \sigma(x),\]
  for any subdistribution \(\sigma \in \subdistr(X)\) and any element \(y \in Y\).

  The monad multiplication \(\monadmultiplication_X \colon \subdistr(\subdistr(X)) \to \subdistr(X)\) and the monad unit \(\monadunit_X \colon X \to \subdistr(X)\) are defined analogously to those of \(\distr\) (\Cref{def:finitary-distributions}).
  Explicitly, \(\monadmultiplication\) is defined by
  \[\monadmultiplication_X (x \given p) \defn \sum_{\sigma \in \subdistr(X)} p(\sigma) \cdot \sigma(x);\]
  the monad unit \(\monadunit\) is defined by \(\monadunit_X(x \given x') \defn \dirac{x'}(x)\),
  where \(\dirac{x} \in \distr(X)\) is the Dirac distribution that assigns probability \(1\) to \(x\) and \(0\) to everything else.
\end{definition}

\begin{remark}
  The fact that \(\subdistr\) is a functor and a monad can be seen by the fact that there is a distributive law between the \emph{Maybe} monad \((\maybe)\) with the finitary \emph{distribution} monad \(\distr\) because the category \(\Set\) of sets and functions is distributive.
  This implies that \((\maybe)\) can be lifted to the Kleisli category \(\kleisli{\distr}\) and that there is a distributive law between \(\distr\) and \((\maybe)\).
  Their composition is the finitary subdistribution monad \(\subdistr = \distr(\maybe)\).
  The distributive law \(\distributivelaw \colon \maybe[\distr(-)] \to \distr(\maybe)\) is defined by
  \(\distributivelaw_{X}(\sigma) \defn \extenddomain{\sigma}\) and \(\distributivelaw_{X}(\bot) \defn \dirac{\bot}\),
  where \(\sigma \in \distr(X)\) and \(\extenddomain{\sigma}\) is \(\sigma\) extended to \(\maybe[X]\) by \(\extenddomain{\sigma}(\bot) \defn 0\).
  See the work of Jacobs for details~\cite[Section 4]{jacobs2018probability}.
\end{remark}

We can check that the Kleisli category of the subdistribution monad, \(\kleisli{\subdistr}\), has conditionals.
However, not every map is total, which prevents it from being a Markov category.

\begin{proposition}\label{prop:finitary-subdistributions}
  The Kleisli category of the finitary subdistribution monad, \(\kleisli{\subdistr}\), is a copy-discard category with conditionals.
\end{proposition}

The analogue of \(\kleisli{\subdistr}\) for \(\BorelStoch\) is \(\subBorelStoch\).
A morphism \(f \colon X \to Y\) in \(\subBorelStoch\) represents a stochastic channel that has some probability of failure, i.e. \(f(x)\) is a subprobability measure.

\begin{definition}
  A \emph{subprobability measure} \(p\) on a measurable space \((X,\sigma_{X})\) is a measurable function such that \(p(X) \leq 1\). %
\end{definition}

\begin{definition}[\cite{panangaden1999}]
  The \defining{linksubborelstoch}{category} \(\subBorelStoch\) has standard Borel spaces \((X,\sigma_{X})\), where \(X\) is a set and \(\sigma_{X}\) is a \(\sigma\)-algebra on \(X\), as objects.
  A morphism \(f \colon (X,\sigma_{X}) \to (Y,\sigma_{Y})\) is a function \(f \colon \sigma_{Y} \times X \to [0,1]\) such that, for each \(T \in \sigma_{Y}\), \(f(T \given -) \colon X \to [0,1]\) is a measurable function, and, for each \(x \in X\), \(f(- \given x) \colon \sigma_{Y} \to [0,1]\) is a subprobability measure.
  Identities and composition are defined analogously to those in \(\BorelStoch\).
\end{definition}

\begin{remark}\label{rem:subBorelStoch-kleisli-cat}
  The category \(\subBorelStoch\) arises as the Kleisli category of \defining{linksubgiry}{Panangaden's monad}~\cite{panangaden1999} \(\subGiry\).
  Its underlying functor is the composition of the Giry functor and the Maybe functor: \(\subGiry \defn \Giry(\maybe)\).
  There is a candidate distributive law \(\distributivelaw \colon \maybe[\Giry(-)] \to \Giry(\maybe)\) defined by \(\distributivelaw_{X}(\sigma) \defn \extenddomain{\sigma}\) and \(\distributivelaw_{X}(\bot) \defn \dirac{\bot}\), where \(\extenddomain{\sigma}\) is the extension of \(\sigma\) to \(\maybe[X]\) by \(\extenddomain{\sigma}(\{\bot\}) \defn 0\).
  \Cref{th:distributive-law-giry-maybe} shows that this is indeed a distributive law between the Giry monad and the Maybe monad.
\end{remark}

\begin{proposition}\label{th:distributive-law-giry-maybe}
  There is a distributive law between the Giry monad and the Maybe monad: \(\distributivelaw \colon \maybe[\Giry(-)] \to \Giry(\maybe)\).
\end{proposition}
\begin{proof}
  The composite functor \(\Giry(\maybe) = \subGiry\) is a monad~\cite{panangaden1999}, with multiplication and unit given by compositions of the multiplications and units of \(\Giry\) and \((\maybe)\).
  For the units, this is easy to see as the unit of \(\Giry(\maybe)\) is just the inclusion of the unit of \(\Giry\), and the inclusion \(\Giry \to \Giry(\maybe)\) is given by the unit of \((\maybe)\).
  For the multiplications, we can check that \(\monadmultiplication = (\id{} \tensor \distributivelaw \tensor \id{}) \dcomp (\monadmultiplication_{1} \tensor \monadmultiplication_{2})\).
  In fact, let \(p \in \Giry(\maybe[\Giry(\maybe[X])])\).
  Then,
  \begin{align*}
    & p \quad \overset{(\id{} \tensor \distributivelaw_{X} \tensor \id{})}{\mapsto}\quad p(\extenddomain{(-)}) \quad \overset{\monadmultiplication_{1} \tensor \monadmultiplication_{2}}{\mapsto} \quad \int_{\tau \in \Giry(\maybe[X])} \tau(-) \cdot p(d\tau) \ ,
  \end{align*}
  which corresponds with the definition of the multiplication of \(\Giry(\maybe)\).
  The components defined in \Cref{rem:subBorelStoch-kleisli-cat} form a natural transformation. %
  These conditions already imply that there is in fact a distributive law between \(\Giry\) and \((\maybe)\).%
\end{proof}

\begin{remark}\label{rem:probability-of-failure}
  In \(\kleisli{\subdistr}\), the composition of a morphism \(f \colon X \to Y\) with the discard map \(\discard\) defines the validity, or probability of success of \(f\): \(f \dcomp \discard \ (\ast \given x) = \sum_{y \in Y} f(y \given x) = 1 - f(\bot \given x)\).
  If the probability of success \(f \dcomp \discard\) is deterministic (\Cref{fig:domain-of-definition}), then \(f(- \given x)\) can either certainly fail, \(f(\bot \given x) = 1\), or give a total distribution \(f(- \given x) \in \distr(Y)\).
  This means that \(f\) factors through the inclusion \(\distr(Y)+1 \into \subdistr\) and \(f \dcomp \discard\) corresponds to the domain of definition of \(f\).

  Following~\cite{cockett02}, we keep this nomenclature in any copy-discard category: we call the morphism \(f \dcomp \discard \colon X \to \monoidalunit\) the \emph{probability of success} of \(f\), and, when it is deterministic, we call it the \emph{domain of definition} of \(f\).
  \begin{figure}[h!]
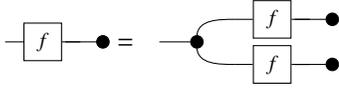

    \[\domainofdefinitionFig{}\]
    \caption{When the probability of success of \(f\) is deterministic, it gives the domain of definition of \(f\).}\label{fig:domain-of-definition}
  \end{figure}
\end{remark}

The categories \(\kleisli{\subdistr}\) and \(\subBorelStoch\) are copy-discard categories with conditionals in which not all morphisms are total.
This means that they cannot be Markov categories.
We claim that totality is not essential for modelling stochastic processes and that dropping this assumption allows us to model \emph{observations}.
 \section{Partial Markov categories}\label{sec:partial-markov}
Cartesian restriction categories extend cartesian categories encoding \emph{partiality}: a map may not be defined on all its inputs and fail when evaluated on inputs outside its domain of definition.
We introduce \emph{partial Markov categories} as a similar extension of Markov categories to encode partial stochastic processes, i.e.\ stochastic processes that have a probability of success on each one of their inputs.
Partiality is obtained by dropping naturality of the discard maps, i.e.\ by allowing morphisms to be non-total.

In a Markov category all morphisms, and in particular all conditionals, are total.
This fact implies that the marginal of a morphism \(f \colon X \to A \tensor B\) on \(A\) is obtained by discarding the \(B\) output.
In partial Markov categories, we would like to drop the totality assumption while still obtaining marginals by discarding one of the outputs.
We could be tempted to impose that conditionals be total morphisms, but this is not possible even in \(\kleisli{\subdistr}\): the ``always fail'' map \(\bot \colon A \to 0\) cannot have a total conditional.
Therefore, we impose a weaker condition which we call \emph{quasi-totality}.

\begin{definition}\label{def:quasi-total-morphism}
  A morphism \(f \colon X \to Y\) is \emph{quasi-total} if \(f = \cp \dcomp (f \tensor (f \dcomp \discard))\).
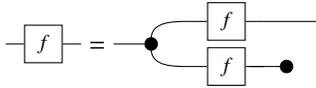
  \[\quasitotalmapdefFig{}\]
\end{definition}

All deterministic morphisms are quasi-total and their probability of success (\Cref{rem:probability-of-failure}) gives their domain of definition as in~\cite{cockett02}.
\Cref{prop:characterisation-quasi-total} precises the relationship between quasi-totality and the domain of definition.

\begin{definition}\label{def:partial-markov-cat}
  A \emph{partial Markov category} is a copy-discard category with quasi-total conditionals.
\end{definition}

\begin{example}
  Any Markov category \(\cat{C}\) is an example of a partial Markov category.
  In fact, a Markov category is a partial Markov category where all maps are total.
\end{example}

\begin{example}
  By \Cref{prop:finitary-subdistributions}, the category \(\kleisli{\subdistr}\) is a partial Markov category because the conditionals defined in its proof are quasi-total.
\end{example}

\begin{proposition}\label{prop:characterisation-quasi-total}
  In a partial Markov category, a morphism \(f \colon X \to Y\) is quasi-total if and only if its probability of success, \(f \dcomp \discard\), is deterministic (\Cref{fig:domain-of-definition}).
\end{proposition}

\begin{example}
  By instantiating \Cref{prop:characterisation-quasi-total} in \(\kleisli{\subdistr}\) and by the considerations in \Cref{rem:probability-of-failure}, a morphism \(f \colon X \to Y\) is quasi-total if and only if it factors through the inclusion \(\distr(X)+1 \into \subdistr(X)\).
\end{example}

After proving \Cref{prop:finitary-subdistributions}, one can conjecture that a similar procedure may exist in Kleisli categories of Maybe monads on other Markov categories.
We show that this is indeed the case.

\subsection{Kleisli categories of Maybe monads}
We show that Kleisli categories of Maybe monads on Markov categories have conditionals inherited by the base category.
This gives a recipe for constructing partial Markov categories from Markov categories with coproducts.
We first show the result for a monad \(\monad\), satisfying some extra conditions, on a Markov category.
These technical conditions will, indeed, be satisfied by the Maybe monad (\Cref{lemma:maybe-monad-satisfies-assumptions}).

\begin{proposition}\label{prop:partial-markov-from-markov}
  Let \(\monad \colon \cat{C} \to \cat{C}\) be a monoidal monad on a copy-discard category with conditionals \(\cat{C}\).
  Suppose that the structural transformation \(\laxator\) is a split epimorphism with section \(\laxatorsec\), that \(\monad \cp_{A} \dcomp \laxatorsec_{A,A} = \cp_{\monad A}\) and that, for every morphism \(f \colon \monad\monad(A) \tensor X \to \monad B\), \((\cp_{\monad A} \tensor \id{X}) \dcomp (\id{\monad A} \tensor \monad \monadunit_{A} \tensor \id{X}) \dcomp (\id{\monad A} \tensor f) \dcomp \monadmultiplication_{A,B} = (\cp_{\monad A} \tensor \id{X}) \dcomp (\id{\monad A} \tensor \monadunit_{\monad A} \tensor \id{X}) \dcomp (\id{\monad A} \tensor f) \dcomp \monadmultiplication_{A,B}\).
  Then, the Kleisli category \(\kleisli{\monad}\) has conditionals.
\end{proposition}
\begin{proof}[Proof sketch]
  Since \(\cat{C}\) has conditionals, there are a marginal \(m \colon A \to \monad A\) and a conditional \(c \colon \monad A \tensor X \to \monad B\) such that \(f \dcomp \laxatorsec_{A,B} = m \condcomp c\), where \(m \condcomp c\) is defined in \Cref{rem:conditional-composition} and \Cref{fig:conditionals}.
  We want to find a marginal \(m' \colon X \to \monad A\) and a conditional \(c' \colon A \tensor X \to \monad B\) such that \(m' \condcomp c' = f\) in \(\kleisli{\monad}\).
  Good candidates for the marginal and the conditional are \(m' \defn m\) and \(c' \defn (\monadunit_A \tensor \id{X}) \dcomp c\).
\end{proof}

The conditionals defined in the proof of \Cref{prop:partial-markov-from-markov} are not necessarily quasi-total.
If the conditionals in the base category \(\cat{C}\) satisfy some additional assumption, we show that conditionals can be chosen to be quasi-total.

\begin{lemma}\label{lemma:quasi-total-conditionals}
  Let \(\monad \colon \cat{C} \to \cat{C}\) be a monoidal monad on a Markov category satisfying the assumptions of \Cref{prop:partial-markov-from-markov}.
  If, for any morphism \(f \colon X \to A \tensor \monad B\) in \(\cat{C}\), we can choose a conditional \(c \colon A \tensor X \to \monad B\) such that \(\cp \dcomp (c \tensor (c \dcomp \monad \discard)) \dcomp \laxator = c \), then the conditionals defined in \Cref{prop:partial-markov-from-markov} are quasi-total.
\end{lemma}
\begin{proof}[Proof sketch]
  By spelling out the definition of quasi-total morphism (\Cref{def:quasi-total-morphism}) in \(\kleisli{\monad}\), we show that it becomes the condition in the statement.
\end{proof}

We now show that \Cref{prop:partial-markov-from-markov} holds for the Maybe monad.
The proof relies on the following result from~\cite{fritz_2020}.

\begin{lemma}[{\cite[Remark 11.29]{fritz_2020}}]\label{prop:markov-comonoid-is-copy}
  In a Markov category, any comonoid coincides with the copy-discard structure.
\end{lemma}

\begin{lemma}\label{lemma:maybe-monad-satisfies-assumptions}
  The Maybe monad \((\maybe) \colon \cat{C} \to \cat{C}\) on a Markov category with coproducts satisfies the assumptions of \Cref{prop:partial-markov-from-markov}.
\end{lemma}
\begin{proof}[Proof sketch]
  Recall that the structural transformation \(\laxator\) for the Maybe monad is defined by \(\laxator_{A,B} \defn \id{A \tensor B} + \coproductmap{\finmap{A + B}}{\id{1}}\).
  We define a candidate for the section \(\laxatorsec\) of the structural transformation \(\laxator\) as \(\laxatorsec_{A,B} \defn \id{A \tensor B} + \initmap{A+B} + \id{1}\).
  Note that \(\laxatorsec\) is a natural transformation because it is a coproduct of natural transformations.
  With this definition, \(\laxatorsec\) is a section of \(\laxator\).
  For the second assumption, we need to show that \((\cp_{A} +1) \dcomp \laxatorsec_{A,A}\) coincides with the copy map \(\cp_{A+1}\).
  By \Cref{prop:markov-comonoid-is-copy}, it suffices to realize that \((\cp_{A} +1) \dcomp \laxatorsec_{A,A}\) forms a comonoid with unit \(\discard_{A+1}\).
  For the third assumption, we rewrite the two sides of the desired equation in normal form, using distributivity of \(\tensor\) over \(+\), and unitality of the universal map from the coproduct \(\cocopy_{A} \colon A + A \to A\).
  We notice that the two sides of the equation are equal because of terminality of \(1\).
\end{proof}

Combining \Cref{prop:partial-markov-from-markov} and \Cref{lemma:maybe-monad-satisfies-assumptions}, we obtain the desired result.
It follows that \(\subBorelStoch\) is a partial Markov category.

\begin{lemma}\label{lemma:borel-quasi-total}
  The Markov category \(\BorelStoch\) satisfies the assumptions of \Cref{lemma:quasi-total-conditionals}.
\end{lemma}

\begin{theorem}\label{th:maybe-subdistributions}
  Let \((\maybe) \colon \cat{C} \to \cat{C}\) be the \emph{Maybe} monad on a Markov category \(\cat{C}\) with coproducts.
  Then, its Kleisli category \(\kleisli{\maybe}\) has conditionals.
  Suppose that \(\cat{C}\) and its \emph{Maybe} monad additionally satisfy the conditions for \Cref{lemma:quasi-total-conditionals}.
  Then, conditionals are quasi-total and \(\kleisli{\maybe}\) is a partial Markov category.
\end{theorem}
\begin{proof}
  The copy-discard structure of \(\kleisli{\maybe}\) is inherited from the base category \(\cat{C}\).
  By \Cref{lemma:maybe-monad-satisfies-assumptions}, the \emph{Maybe} monad on a Markov category with coproducts satisfies the assumptions of \Cref{prop:partial-markov-from-markov}.
  This shows that \(\kleisli{\maybe}\) has conditionals and, by \Cref{lemma:quasi-total-conditionals}, they are quasi-total, which makes \(\kleisli{\maybe}\) a partial Markov category.
\end{proof}

\begin{corollary}\label{prop:borel-subdistributions}
  \(\subBorelStoch\) is a partial Markov category.
\end{corollary}
\begin{proof}
  By \Cref{th:distributive-law-giry-maybe}, the category \(\subBorelStoch\) can be seen as the Kleisli category of the Maybe monad \((\maybe) \colon \BorelStoch \to \BorelStoch\) on standard Borel spaces.
  \(\BorelStoch\) is a Markov category and satisfies the assumptions for \Cref{lemma:quasi-total-conditionals} by \Cref{lemma:borel-quasi-total}.
  Then, we can apply \Cref{th:maybe-subdistributions} to obtain that \(\subBorelStoch\) is a partial Markov category.
\end{proof}

\subsection{Some properties of partial Markov categories}

A first useful property of Markov categories is that the marginal on \(A\) of a morphism \(f \colon X \to A \tensor B\) must be given by discarding its \(B\) output: \(f \dcomp (\id{A} \tensor \discard_{B})\).
This characterisation relies on the fact that all morphisms are total, which is not the case in a partial Markov category~\cite{fritz_2020}.
However, the fact that conditionals are quasi-total is enough to give an equivalent characterisation of marginals.
We will implicitly apply this result in all the following sections.

\begin{proposition}\label{prop:QTconditionals-give-nice-marginals}
  Let \(f \colon X \to A \tensor B\) be a morphism in a partial Markov category \(\cat{C}\) with a quasi-total conditional \(c \colon A \tensor X \to B\).
  Then, \(f \dcomp (\id{A} \tensor \discard_{B})\) is a marginal associated with \(c\).
\end{proposition}
\begin{proof}
  We employ string diagrammatic reasoning.
  \marginalQTconditionalsProofFig{}
  \Cref{eq:nice-conditionals-conditional1,eq:nice-conditionals-conditional2} follow from the assumption that \(m\) and \(c\) are a marginal and a conditional of \(f\), \Cref{eq:nice-conditionals-associativity} is an application of associativity of the copy-discard structure, and \Cref{eq:nice-conditionals-quasi-total} follows from quasi-totality of \(c\) and uniformity of the comonoid structure.
\end{proof}

In Markov categories, if conditionals are unique, the category collapses to a preorder~\cite[Proposition 11.15]{fritz_2020}.
For partial Markov categories, the collapse is not so extreme but, when conditionals are unique, parallel morphisms are characterised by their probability of success.

\begin{proposition}\label{prop:unique-conditionals-collapse}
  Let \(\cat{C}\) be a partial Markov category and suppose that all conditionals are unique.
  Then, for every two morphisms \(f, g \colon X \to Y\), if \(f \dcomp \discard = g \dcomp \discard\) then \(f=g\).
\end{proposition}
\begin{proof}[Proof sketch]
  The proof is extended from~\cite[Proposition 11.15]{fritz_2020} to rely on quasi-totality instead of totality.
\end{proof}

\subsection{Bayesian inversion and normalisation}
The \emph{Bayesian inversion} of a stochastic channel \(g \colon X \to Y\) with respect to a distribution \(\sigma\) over \(X\) is the stochastic channel \(\bayesinv{g}{\sigma} \colon Y \to X\) classically defined as below\footnote{Bayesian inversions are uniquely defined for all \(y \in Y\) with positive probability, \(\sum_{x_{\bullet}\in X}g(y \given x_{\bullet}) \sigma(x_{\bullet}) > 0\)}.
\[\bayesinv{g}{\sigma}(x \given y) = \dfrac{g(y \given x)\sigma(x)}{\sum_{x_\bullet \in X} g(y \given x_\bullet)\sigma(x_\bullet)}\]
Bayesian inversions can be defined abstractly in partial Markov categories, as they can be in Markov categories~\cite[Proposition 11.17]{fritz_2020}.
Bayesian inversions are just a particular case of conditionals.
We state this result for partial Markov categories (\Cref{prop:bayes-inversions-from-conditionals}) as a straightforward generalisation of~\cite[Proposition 11.17]{fritz_2020}.

\begin{definition}\label{def:bayes-inversion}
  A \defining{linkbayesinv}{\emph{Bayesian inversion}} of a morphism \(g \colon X \to Y\) with respect to \(\sigma \colon \monoidalunit \to X\) is a morphism \(\bayesinv{g}{\sigma} \colon Y \to X\) %
  satisfying the equation in \Cref{diagram-bayesian-inversion}.
  \begin{figure}[h!]
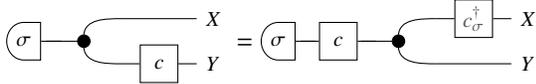

    \BayesinversionDefFig{}
    \caption{Bayesian inversion.}\label{diagram-bayesian-inversion}
  \end{figure}
\end{definition}

\begin{proposition}\label{prop:bayes-inversions-from-conditionals}
  In a partial Markov category, all Bayesian inversions exist.
\end{proposition}
\begin{proof}
  This can be easily checked by applying the axiom of conditionals to the morphism defined by \(\sigma \comp \cp_{X} \comp (g \tensor \id{})\).
\end{proof}

Bayesian inversions can be computed compositionally.
Similar results appeared in~\cite[Lemma 11.11]{fritz_2020} for Markov categories and~\cite[Section 5.1]{jacobs_2019} for \(\kleisli{\distr}\).
We recast it in the setting of partial Markov categories to prove \Cref{prop:pearlandjeffrey}.

\begin{proposition}\label{prop:composition-bayes-inversions}
  A Bayes inversion of a composite channel \(f \defn c \dcomp d \colon X \to Y\) with respect to a state \(\sigma \colon \monoidalunit \to X\) can be computed by first inverting \(c\) with respect to \(\sigma\) and then inverting \(d\) with respect to \(\sigma \dcomp c\):
  \[\bayesinv{(c \dcomp d)}{\sigma} = \bayesinv{d}{\sigma \dcomp c} \dcomp \bayesinv{c}{\sigma}.\]
\end{proposition}

The normalisation of a partial stochastic channel \(f \colon X \to Y\) is classically defined as\footnote{Normalisations are uniquely defined for all \(x \in X\) with probability of failure \(<1\), i.e. \(1 - f(\bot \given x) > 0\)}
\[\normal{f} (- \given x) \defn \dfrac{f(- \given x)}{1- f(\bot \given x)}.\]
Normalisations can be defined in any partial Markov category.
In Markov categories, this notion trivialises as all morphisms are required to be total.

\begin{definition}\label{def:normalisation}
  Let \(f \colon X \to Y\) be a morphism in a partial Markov category.
  A \defining{linknormal}{\emph{normalisation}} of \(f\) is a quasi-total morphism \(\normal{f} \colon X \to Y\) such that
  \(f = \cp \dcomp (\normal{f} \tensor (f \dcomp \discard))\).
  \stochasticchannelquasinormalisationFig{}
\end{definition}

Note that the quasi-totality requirement in \Cref{eq:normalisation-idempotent} is \Cref{eq:def-normalisation} for \(f = \normal{f}\).
This ensures that normalising is an idempotent operation, i.e.~that \(\normal{\normal{f}} = \normal{f}\).
Normalisations are also a particular case of conditionals.

\begin{proposition}
  In a partial Markov category, all normalisations exist.
\end{proposition}
\begin{proof}
  The normalisation of \(f\) is given by its quasi-total conditional on the monoidal unit wire.
\end{proof}

Normalisation does not influence conditioning.
\Cref{prop:conditionals-of-programs} relies on this result.

\begin{proposition}\label{prop:conditionals-of-normalisation}
  Let \(f \colon X \to Y \tensor Z\) be a morphism in a partial Markov category.
  Then, the conditionals of the normalisation of \(f\) are conditionals of \(f\).
\end{proposition}
\begin{proof}
  We employ string diagrams.
  \conditionalsfromnormalisationProofFig{}
\end{proof}

\subsection{Discrete partial Markov categories for Bayes update}
We introduce \emph{discrete partial Markov categories}, a refinement of partial Markov categories that allows for the encoding of \emph{constraints}.

Discrete cartesian restriction categories~\cite{cockett2012range2} are a refinement of cartesian restriction categories that allows for the encoding of \emph{constraints}: a map may fail if some conditions are not satisfied.
Our observation is that a similar refinement can be applied to partial Markov categories to obtain discrete partial Markov categories.
They provide a setting in which it is possible to (i) constrain, via Bayesian updates; and (ii) reason with stochastic maps.

The encoding of constraints requires the existence of \emph{comparator maps} %
that interact nicely with the copy-discard structure (\Cref{diagram-partial-frobenius}, see also \Cref{rem:reading-partial-frobenius}).
A comparator declares that some constraint --- usually an observation, on which we condition --- needs to be satisfied in a probabilistic process.

\begin{definition}
  A copy-discard category \(\cat{C}\) has \defining{linkcompare}{\emph{comparators}} if every object \(X\) has a morphism \(\compare_X \colon X \tensor X \to X\) that is uniform, commutative, associative and satisfies the Frobenius axioms with the copy-discard structure, as in \Cref{diagram-partial-frobenius}.
\end{definition}
\begin{figure}[h!]
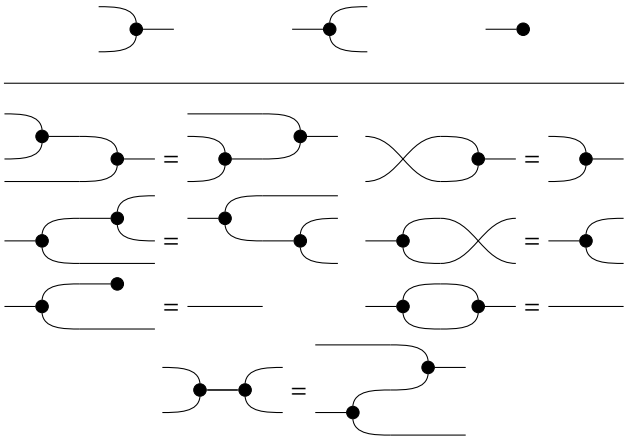

  \partialFrobunniusFig{}
  \caption{Axioms of a partial Frobenius monoid.}\label{diagram-partial-frobenius}
\end{figure}

\begin{definition}\label{def:discrete-partial-markov-cat}
  A \emph{discrete partial Markov category} is a copy-discard category with conditionals and comparators. In other words, it is a partial Markov category with comparators.
\end{definition}

\begin{example}
  The Kleisli category of the finitary subdistribution monad, \(\kleisli{\subdistr}\), is a discrete partial Markov category.
  The comparator \(\compare_{X} \colon X \tensor X \to X\) is given by
  \[\compare_{X}(x \given x_{1},x_{2}) \defn \begin{cases} 1, & x=x_{1}=x_{2}; \\ 0, & \text{otherwise}.\end{cases}\]
  This morphism and the copy-discard structure of \(\kleisli{\subdistr}\) are lifted from the category of partial functions via the inclusion \(\inclusion \colon \Par \into \kleisli{\subdistr}\) given by post-composition with the unit of the finitary distribution monad.
  The comparator in \(\Par\) satisfies the axioms in \Cref{diagram-partial-frobenius} as copying a resource and then checking that the two copies coincide should be the identity process, and for checking equality of two resources and then copying them should be the same as copying one resource if it coincides with the other one.
  By functoriality of the inclusion \(\inclusion \colon \Par \into \kleisli{\subdistr}\), the comparator satisfies the same axioms in \(\kleisli{\subdistr}\).
\end{example}

\begin{remark}\label{rem:reading-partial-frobenius}
  Thanks to the special Frobenius axioms (\Cref{diagram-partial-frobenius}), string diagrams in \(\kleisli{\subdistr}\) keep the same intuitive reading as in \(\kleisli{\distr}\): the value of a morphism is obtained by multiplying the values of all its components and summing on the wires that are not inputs nor outputs.
  For example, the value of the morphism below is
  \[f(z \given x) = \sum_{y \in Y} g(y \given x) \cdot \sigma(y) \cdot h(z \given y).\]
  \[\morphismdiscretecopydiscardExampleFig{}\]
\end{remark}

\begin{example}\label{ex:borel-subdistributions-no-comparator}
  The category \(\subBorelStoch\) has a comparator \(\compare_{X} \colon (X, \sigma_{X}) \tensor (X, \sigma_{X}) \to (X, \sigma_{X})\) defined so that it satisfies the special Frobenius axioms:
    \[\compare_{X}(A, x, y) =
      \begin{cases}
        1, & \mbox{ if } x=y \in A; \\
        0, & \text{otherwise}.
      \end{cases}\]
  This definition gives a measurable function \(\compare_{X}(A \given -,-) \colon X \times X \to [0,1]\) if and only if the diagonal \(\Delta_{X} = \{(x,x) \st x \in X\}\) belongs to the product \(\sigma\)-algebra \(\sigma_{X} \times \sigma_{X}\), which is true in standard Borel spaces.
  However, this naive comparator does not behave in the way we would like: the set \(\{x_{0}\}\) has measure \(0\), so comparing with \(x_{0}\) yields the subdistribution with measure \(0\), which cannot be renormalised.
\end{example}

\subsection{Bayes' Theorem}

Bayes' theorem prescribes how to update one's belief in light of new evidence.
Classically, one observes evidence \(y \in Y\) from a prior distribution \(\sigma\) on \(X\) through a channel \(c \colon X \to Y\).
The updated distribution is given by evaluating the Bayesian inversion of the channel \(c\) on the new observation \(y\).

\begin{theorem}[Synthetic Bayes' Theorem]\label{th:bayes}
  In a discrete partial Markov category, observing a deterministic \(y \colon \monoidalunit \to Y\) from a prior distribution \(\sigma \colon \monoidalunit \to X\) through a channel \(c \colon X \to Y\) is the same, up to scalar, as evaluating the Bayesian inversion of the channel \(\bayesinv{c}{\sigma}(y)\).
  \[\bayesthstatementFig{}\]
\end{theorem}
\begin{proof}
  We employ string diagrams.%
  \bayesthProofFig{}

  The equalities follow from: (\emph{i}) the definition of Bayesian inversion (\Cref{def:bayes-inversion}) and \Cref{prop:QTconditionals-give-nice-marginals}, (\emph{ii}) the partial Frobenius axioms (\Cref{diagram-partial-frobenius}), and (\emph{iii}) determinism of the observation $y$.
\end{proof}

\subsection{Pearl's and Jeffrey's updates}
The process for updating a belief on new evidence may depend on the type of evidence given.
Pearl's~\cite{pearl1988probabilistic,pearl1990jeffrey} and Jeffrey's~\cite{jeffrey1990logic,shafer1981jeffrey,halpern2017reasoning} updates are two possibilities for performing an update of a belief in light of new evidence.
Updating a prior belief according to Pearl's rule increases \emph{validity}, i.e.\ the probability of the new evidence being true according to our belief~\cite{cho2015introduction}.
On the other hand, updating with Jeffrey's rule reduces ``how far'' the new evidence is from our prediction, i.e.\ it decreases \emph{Kullback-Leibler divergence}~\cite{jacobs_2019,jacobs2021learning}.

The difference between these two update rules comes from the fact that they are based on different types of evidence.
Pearl's evidence comes as a probabilistic predicate, i.e.\ a morphism \(q \colon Y \to \monoidalunit\) in a discrete partial Markov category. Pearl's update coincides with the update prescribed by Bayes.

\begin{definition}\label{def:pearl-update}
  Let \(\sigma \colon \monoidalunit \to X\) be a prior distribution and \(q \colon Y \to \monoidalunit\) be a predicate in a discrete partial Markov category \(\cat{C}\), which is observed through a channel \(c \colon X \to Y\).
  Pearl's updated prior is defined to be \(\bayesinv{(c \dcomp q)}{\sigma}\), the total Bayes inversion of \(p \defn c \dcomp q\) with respect to \(\sigma\):
  \[\pearlupdatedefFig\]
\end{definition}

Jeffrey's evidence, on the other hand, is given by a distribution on \(Y\).

\begin{definition}\label{def:jeffrey-update}
  Let \(\tau \colon \monoidalunit \to Y\) be a state in \(\cat{C}\).
  \emph{Jeffrey's updated prior} is \(\tau \dcomp \bayesinv{c}{\sigma}\), the composition of the evidence with the total Bayes inversion of \(c\) with respect to \(\sigma\).
  \[\jeffreyupdatedefFig\]
\end{definition}

When Pearl's evidence predicate \(q\) is deterministic, which means that its probability mass is concentrated in just one point \(y \in Y\), then \(q\) can be written as a constraint.
\begin{equation}\label{eq:map-jeffrey-to-pearl}
  \pearlevidencedeterministicFig{}
\end{equation}
In this case, there is no difference between the two update rules.
This result was proven in~\cite[Proposition 5.3]{jacobs_2019} in the case of the Kleisli category of finitary distribution monad.
We prove it in any discrete partial Markov category.

\begin{proposition}\label{prop:pearlandjeffrey}
  If \(y \colon \monoidalunit \to Y\) is deterministic, then Pearl's update on the predicate \(q \colon Y \to \monoidalunit\), as defined in \Cref{eq:map-jeffrey-to-pearl},
  is Jeffrey's update on \(y\).
\end{proposition}
\begin{proof}
  The result follows from \Cref{prop:composition-bayes-inversions}, \Cref{prop:QTconditionals-give-nice-marginals} and the partial Frobenius axioms (\Cref{diagram-partial-frobenius}), by a string diagrammatic reasoning similar to that of the proof of \Cref{th:bayes}.
\end{proof}

 \section{Evidential Decision Theory}\label{sec:evidential-dec-theory}
This section aims to model decision problems as morphisms in a free partial Markov category and depict them with string diagrams.
Providing semantics for each node of the string diagram (e.g. in terms of subdistributions) will automatically induce semantics for the whole model.
The question that Evidential Decision Theory aims to answer is:
\slogan{Which action is evidence for the best possible outcome?}
This means that the optimal answer to the problem will be the one that, once observed as the output of %
the agent's node, maximises the outcome.
This section expresses this question in terms of the calculus of discrete partial Markov categories.%

\subsection{Solving decision problems in partial Markov categories}
\label{sec:solvingdecision}
The decision problems that we introduce in \Cref{sec:decisionproblems} are defined by three elements: \emph{(i)} an environment, \(w \colon \monoidalunit \to C \tensor O\); \emph{(ii)} an agent,  \(a \colon O \to A\), that observes a part of the enviroment (\(O\)) and chooses an action,  \(x \in A\); \emph{(iii)} and a partial stochastic process \(f \colon C \tensor A \to U\) that imposes some ``compatibility'' conditions on the action and computes the utility.
The following diagram is an abstract model of a decision problem. 
\[\monoidalcontextdecisionproblemdefFig{}\]

Evidential decision theory prescribes the choice of the action \(x \in A\) that we would observe in case we obtained the maximum (average) utility.
We translate this statement in the formalism of discrete partial Markov categories: evidential decision theory prescribes the action \(x \in A\) such that the corresponding deterministic state \(x \colon \monoidalunit \to A\) maximises the average of the normalisation of the following subdistribution over utilities.
\[\observingactionmonoidalcontextdecisionproblemFig{}\]

By applying the Frobenius axioms (\Cref{diagram-partial-frobenius}) in (\ref{eq:dec-problem-frob1}) and (\ref{eq:dec-problem-frob2}), \Cref{prop:QTconditionals-give-nice-marginals} in (\ref{eq:dec-problem-condition1}) and (\ref{eq:dec-problem-condition2}), and the fact that the choice of strategy is deterministic in (\ref{eq:dec-problem-det1}) and (\ref{eq:dec-problem-det2}), we can always reduce this model.
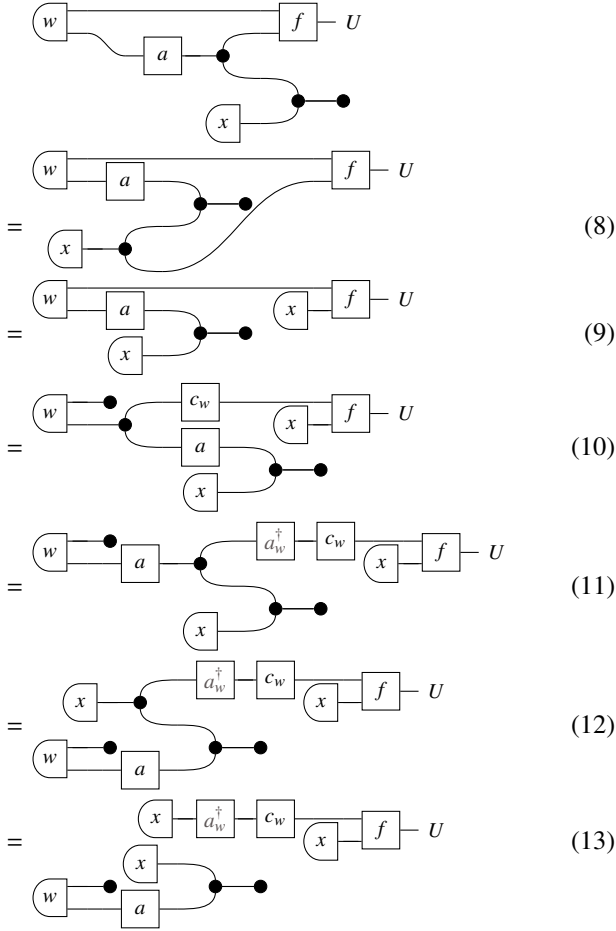
\solvingmonoidalcontextdecisionproblemFig{}

Indeed, in all the examples that we will present next, we only perform comparisons with deterministic states: in other words, the only observations that we need to encode are \emph{deterministic}.
This motivates the results of \Cref{sec:observations-conditionals}: it is possible to add \emph{deterministic} observations syntactically to any Markov category and reduce the computations to the original Markov category.
This allows to reason with deterministic observations even in categories without a comparator structure.

\subsection{Some decision problems}
\label{sec:decisionproblems}
We start by a classical example of decision problem: the Monty-Hall problem~\cite{vosSavant}.
Later, we study Newcomb problem~\cite{nozick_1969} %
this is a problem where Evidential Decision Theory prescribes the action that maximises utility, in contrast with Causal Decision Theory.
Finally, we model the Smoking-lesion problem, which exemplifies the class of problems where Causal Decision Theory outperforms Evidential Decision Theory.
Whenever Causal Decision Theory outperforms Evidential Decision Theory, the categorical modelling makes explicit the assumptions that lead to their discrepancy.
This allows us to clarify which problems satisfy these assumptions.

\paragraph*{Monty-Hall Problem}
An agent is in front of three doors (\(D\)) and is given the choice to pick one of them (\(door \in \distr(D)\)).
Behind one of these doors there is a prize (\(prize \in \distr(D)\)), say of \(1000 \utility\), while behind the other ones there is a goat, which corresponds to \(0 \utility\).
\begin{center}
  \begin{tabular}{| c  c |}
    \hline
    Outcomes & Utility \\
    \hline
    \(open = prize\) & \(1000 \utility\) \\
    \(open \neq prize\) & \(0 \utility\)\\
    \hline
  \end{tabular}
\end{center}
A predictor knows which door hides the prize and opens one of the doors that does not.
The agent is confronted with two choices: keep the original choice of door or change the choice.
Which action should the agent choose?

\begin{figure}[h!]
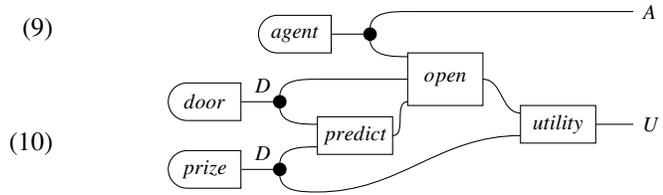

  \[\montyhallFig{}\]
  \caption{Model for Monty-Hall problem.}
  \label{fig:montyhallModel}
\end{figure}
In this classical example, both Evidential and Causal Decision Theory prescribe changing doors.  Indeed, that is the observation that, attached to the first output of \Cref{fig:montyhallModel}, maximises expected utility on the second output.

\paragraph*{Newcomb's Problem}\label{sec:newcomb}
We are now able to model Newcomb's paradox as stated in \Cref{sec:introduction} (\Cref{fig:newcomb-utilities}) and interpret the model in \Cref{fig:newcomb-diagram} in the discrete partial Markov category \(\kleisli{\subdistr}\).

\paragraph*{Transparent Newcomb's problem}
An agent confronts Newcomb's problem with the only difference that both boxes are transparent, which allows the agent to observe which prediction was made.
Should the agent \emph{one-box} or \emph{two-box}?
\begin{figure}[h!]
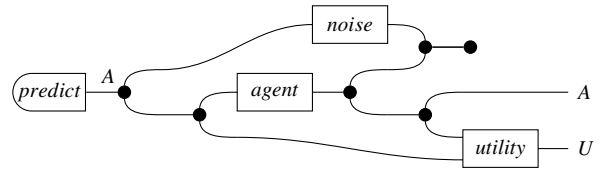

  \[\transparentnewcombFig{}\]
  \caption{Model for the Transparent Newcomb problem.}\label{diagram-transparent-newcomb-model}
\end{figure}

Perhaps surprisingly, with the same reasoning, Evidential Decision Theory still prescribes one-boxing.

\paragraph*{Smoking-lesion problem}
In an imaginary world there is a gene that causes both cancer and desire to smoke (\(gene \in \distr(G \times D)\)).
Smoking gives a small utility, say \(1 \utility\), while dying of cancer gives a large negative utility, say \(-1000 \utility\).
\begin{center}
  \begin{tabular}{| c | c  c |}
    \hline
    & cancer & not cancer \\
    \hline
    smoke & \(-999 \utility\) & \(1 \utility\) \\
    not smoke & \(-1000 \utility\) & \(0 \utility\)\\
    \hline
  \end{tabular}
\end{center}
Should an agent decide to smoke or to refrain from smoking?

\begin{figure}[h!]
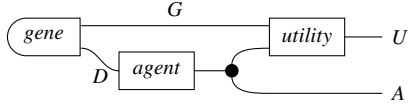

  \[\smokinglesionFig{}\]
  \caption{Model for the Smoking-lesion problem.}
\end{figure}

It is implicit in the statement of the problem that there is a correlation between the desire to smoke and smoking.
Therefore, the choice to smoke is evidence for the gene being present, which, in turn, causes cancer.
This reasoning leads Evidential Decision Theory to prescribe refraining from smoking.
On the other hand, Causal Decision Theory prescribes the action that causes the best utility, which, in this case, is smoking.

 \section{Updating on deterministic observations}\label{sec:observations-conditionals}
Modelling many decision problems in partial Markov categories relies on the existence of comparators.
However, some partial Markov categories may not have this structure or it might not behave as expected (see \Cref{ex:borel-subdistributions-no-comparator}).
The solution to this problem comes from the observation that, in order to solve the decision problems we are interested in, the only observations that appear are deterministic.
We define the category \(\consproc{\cat{C}}\) of \emph{\constrainedProcesses} over a copy-discard category \(\cat{C}\) by syntactically adding observations of deterministic evidence.
A similar construction has appeared in~\cite{stein2021compositional, stein_thesis_2021}.
We show that, when \(\cat{C}\) is a Markov category, constrained processes form a partial Markov category, and that normalisations and updates of them can be computed by taking conditionals only in the original Markov category \(\cat{C}\).

\begin{definition}\label{def:exact-conditioning-cat}
  Let \(\cat{C}\) be a copy-discard category.
  The \defining{linkconsprocC}{category} \(\consproc{\cat{C}}\) of \constrainedProcesses{} in \(\cat{C}\) is obtained by freely adding \(\obsv{y} \colon Y \to \monoidalunit\) for every deterministic map \(y \colon \monoidalunit \to Y\) in \(\cat{C}\) and quotienting by the \defining{linkobsv}{axiom} in \Cref{fig:axiom-exact-conditioning-cat}.
  \begin{figure}[h!]
    \[\axiomconstraintFig{}\]
    \caption{Axiom for the category of \constrainedProcesses{}.}\label{fig:axiom-exact-conditioning-cat}
  \end{figure}
\end{definition}

Intuitively, the generator \(\obsv{y} \colon Y \to \monoidalunit\) corresponds to the observation of the corresponding deterministic evidence \(y \colon \monoidalunit \to Y\) as shown in \Cref{fig:map-deterministic-evidence}.

\begin{proposition}\label{prop:embedding-cons-to-par}
  The category \(\consproc{\cat{C}}\) of \constrainedProcesses{} in \(\cat{C}\) embeds in \(\parproc{\cat{C}}\), the free discrete copy-discard category over \(\cat{C}\).
\end{proposition}
\begin{proof}[Proof sketch]
  We define an identity-on-objects functor \(\fun{F} \colon \consproc{\cat{C}} \to \parproc{\cat{C}}\).
  Every morphism \(f\) in \(\consproc{\cat{C}}\) that comes from a morphism in \(\cat{C}\) is left unchanged: \(\fun{F}(f) \defn f\).
  For every deterministic state \(y \colon 1 \to Y\) in \(\cat{C}\), the image of its corresponding costate in \(\parproc{\cat{C}}\) is defined in \Cref{fig:map-deterministic-evidence}, using the comparator structure.
  \begin{figure}[h!]
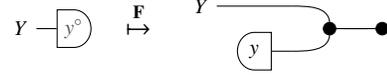

    \[\mappingcostatesConsprocToParprocFig{}\]
    \caption{Correspondence between constrained processes and partial processes.}\label{fig:map-deterministic-evidence}
  \end{figure}
  The fact that \(\fun{F}\) is well defined and its faithfulness follow from string the diagrammatic reasoning. %
\end{proof}

\ConstrainedProcesses{} on a Markov category give a syntax for stochastic processes with some observations of deterministic evidence.
In principle, it is not clear how to compute the semantics of these \constrainedProcesses{} and, in particular, how to compute conditionals of them.
We show that we can give semantics to a \constrainedProcess{} in the original Markov category by computing its normalisation.

\begin{theorem}\label{prop:normalisation-of-programs} Let \(f\) be a
  \constrainedProcess{} in a Markov category \(\cat{C}\). Then, a normalisation
  of \(f\) can be computed by taking conditionals of morphisms in \(\cat{C}\).
  More precisely, \(f\) can be written as below, with \(g\) and \(h\) total, and
  \(z\) total and deterministic.
  \begin{equation}\label{eq:normal-form-programs}
    \normalformprogramsFig{}
  \end{equation}
\end{theorem}
\begin{proof}[Proof sketch]
  We start by observing that any \(g\) in \(\cat{C}\) that satisfies \Cref{eq:normal-form-programs} must be a normalisation of \(f\).
  The proof then proceeds by structural induction on \(f\).
\end{proof}

A consequence of this result is that conditionals of \constrainedProcesses{} can be computed by conditionals in the original Markov category.

\begin{theorem}\label{prop:conditionals-of-programs}
Let \(f \colon X \to Y \tensor Z\) be a \constrainedProcess{} in a Markov category \(\cat{C}\).
Then, \(f\) has conditionals and they can be computed by taking conditionals of morphisms in \(\cat{C}\).
In particular, \(\consproc{\cat{C}}\) is a partial Markov category.
\end{theorem}
\begin{proof}
By \Cref{prop:normalisation-of-programs}, we can compute a normalisation \(\normal{f}\) of a \constrainedProcess{} \(f\) in \(\cat{C}\) by only taking conditionals in \(\cat{C}\).
By \Cref{prop:conditionals-of-normalisation}, a conditional of \(\normal{f}\) is a conditional of \(f\) and it can be computed by only taking conditionals in \(\cat{C}\).
\end{proof}
 \section{Conclusions and further work}
We have introduced partial Markov categories and discrete partial Markov categories as a common extension of Markov categories and (discrete) cartesian restriction categories.
We have shown a synthetic version of Bayes' theorem, which holds in any discrete partial Markov category.
Partial Markov categories provide a good formalism to express Pearl's and Jeffrey's updates, model decision problems like Newcomb's paradox, and solve them according to Evidential Decision Theory.

Programming with exact conditioning~\cite{stein_thesis_2021} relies on being able to express stochastic processes and deterministic observations.
We have shown that discrete partial Markov categories provide a good syntax for observations and updates, but, while we have shown that the Kleisli category of any \emph{Maybe} monad on a Markov category is a partial Markov category, not all of them are \emph{discrete}.
We have solved this issue by defining a construction that freely adds deterministic observations to any Markov category, which provides a syntax for programming with exact conditioning.
We have shown that all the morphisms in this free construction can be normalised by computations in the base Markov category.

\subsection{Further work}
Monoidal categories are a minimalistic framework for process description.
Thus, necessarily, some aspects, traditionally useful in probabilistic programming, are not considered here, e.g.~higher-order functions~\cite{heunen_kammar_staton_yang_2017, vakar2019domain, ehrhard2017measurable} or recursion.
The strength of the minimalistic approach is that it can be extended with these features via, e.g., monoidal closed categories~\cite{dahlqvist19semantics} and traced monoidal categories~\cite{hasegawa97}.
Lastly, an already very promising connection is the proximity of our $\consproc{}$ construction, deduced from partial Markov categories, and Stein's \(\mathbf{Cond}\) construction~\cite{stein2021compositional,stein_thesis_2021} for probabilistic programming with exact conditioning.
These connections can potentially translate, for the first time, between the string diagrammatic approach to decision theory and the syntax of probabilistic programming.

The present manuscript focuses on one of the most important and perhaps less discussed frameworks for decision theory: Evidential Decision Theory.
A comprehensive comparison of the different frameworks in terms of partial Markov categories is left for further work. 
For instance,
Evidential Decision Theory needs careful modelling to solve problems such as the ``Smoking Lesion Problem''~\cite{yudkowsky_soares_2017}.
It is sometimes claimed that these problems are better solved by Causal Decision Theory~\cite{gibbard_harper_1978}, which makes use of ``interventions'' to apply an action to a node of a causal graph~\cite{pearl_2009,jacobs_2019,jacobs2021causal}. These interventions may be still analysed in freely generated Markov or partial Markov categories~\cite{fritz_liang_2022}.

Finally, even when this manuscript deals only with a finite number of updates, we have the tools to study \emph{iterated updates}, which are the basis for results such as Aumann's agreement theorem~\cite{aumann2016agreeing}.
Indeed, iterated probabilistic processes in terms of monoidal categories have been a recent contribution to this Symposium on Logic in Computer Science~\cite{monoidalstreams}.

 \newpage
\subsection*{Acknowledgements}
We thank our supervisor, Paweł Sobociński, and the anonymous referees for the useful comments that improved this manuscript. We gratefully thank Dario Stein for pointing out a mistake in  a previous version of \Cref{ex:borel-subdistributions-no-comparator}. We thank Paolo Perrone for finding a typo where ``probability of failure'' meant ``probability of success''. We thank Márk Széles for pointing out that a previous version of the proof of \Cref{prop:embedding-cons-to-par} used an ``always fail'' map that does not exist in arbitrary partial Markov categories. We thank Siddharth Bhat, Pim de Haan, Miguel Lopez and Ruben Van Belle for discussion on the first versions of this manuscript. Finally, we thank helpful discussions with Vanessa Kosoy, David Darlymple, David Spivak and Scott Garrabrant.

\bibliographystyle{alpha}
\bibliography{bibliography}

\newcommand{\etalchar}[1]{$^{#1}$}
\begin{thebibliography}{DLdFR22}

\bibitem[ACU15]{altenkirch_chapman_uustalu}
Thorsten Altenkirch, James Chapman, and Tarmo Uustalu.
\newblock Monads need not be endofunctors.
\newblock {\em Logical Methods in Computer Science}, 11(1), 2015.

\bibitem[Ahm14]{ahmed_2014}
Arif Ahmed.
\newblock {\em Evidence, {Decision} and {Causality}}.
\newblock Cambridge University Press, 2014.

\bibitem[Aum16]{aumann2016agreeing}
Robert~J Aumann.
\newblock Agreeing to disagree.
\newblock In {\em Readings in Formal Epistemology}, pages 859--862. Springer, 2016.

\bibitem[Car87]{carboni1987bicategories}
Aurelio Carboni.
\newblock Bicategories of partial maps.
\newblock {\em Cahiers de topologie et g{\'e}om{\'e}trie diff{\'e}rentielle cat{\'e}goriques}, 28(2):111--126, 1987.

\bibitem[CG99]{corradini_gadducci_1999}
Andrea Corradini and Fabio Gadducci.
\newblock An {Algebraic} {Presentation} of {Term} {Graphs}, via {GS}-{Monoidal} {Categories}.
\newblock {\em Applied Categorical Structures}, 7(4):299--331, 1999.

\bibitem[CGH12]{cockett2012range2}
JRB Cockett, Xiuzhan Guo, and Pieter Hofstra.
\newblock Range categories ii: Towards regularity.
\newblock {\em Theory and Applications of Categories}, 26(18):453--500, 2012.

\bibitem[CJ19]{cho_2019}
Kenta Cho and Bart Jacobs.
\newblock Disintegration and {{Bayesian Inversion}} via {{String Diagrams}}.
\newblock {\em Mathematical Structures in Computer Science}, pages 1--34, March 2019.

\bibitem[CJWW15]{cho2015introduction}
Kenta Cho, Bart Jacobs, Bas Westerbaan, and Abraham Westerbaan.
\newblock An introduction to effectus theory.
\newblock {\em arXiv preprint arXiv:1512.05813}, 2015.

\bibitem[CL02]{cockett02}
J.~Robin~B. Cockett and Stephen Lack.
\newblock Restriction categories {I:} categories of partial maps.
\newblock {\em Theoretical Computer Science}, 270(1-2):223--259, 2002.

\bibitem[CL03]{cockett2003restriction}
J.~Robin~B. Cockett and Stephen Lack.
\newblock Restriction categories ii: partial map classification.
\newblock {\em Theoretical Computer Science}, 294(1-2):61--102, 2003.

\bibitem[CL07]{cockett2007restriction}
Robin Cockett and Stephen Lack.
\newblock Restriction categories {III}: colimits, partial limits and extensivity.
\newblock {\em Mathematical Structures in Computer Science}, 17(4):775--817, 2007.

\bibitem[CO89]{curien_obtulowicz_1989}
Pierre{-}Louis Curien and Adam Obtu{\l}owicz.
\newblock Partiality, cartesian closedness and toposes.
\newblock {\em Information and Computation}, 80(1):50--95, 1989.

\bibitem[CS12]{coeckespekkens2012picturing}
Bob Coecke and Robert~W Spekkens.
\newblock Picturing classical and quantum bayesian inference.
\newblock {\em Synthese}, 186:651--696, 2012.

\bibitem[DK19]{dahlqvist19semantics}
Fredrik Dahlqvist and Dexter Kozen.
\newblock Semantics of higher-order probabilistic programs with conditioning.
\newblock {\em Proc. ACM Program. Lang.}, 4(POPL), dec 2019.

\bibitem[DLdFR22]{monoidalstreams}
Elena Di~Lavore, Giovanni de~Felice, and Mario Rom\'{a}n.
\newblock Monoidal streams for dataflow programming.
\newblock In {\em Proceedings of the 37th Annual ACM/IEEE Symposium on Logic in Computer Science}, LICS '22, New York, NY, USA, 2022. Association for Computing Machinery.

\bibitem[DLNS21]{di_liberti_nester_2021}
Ivan {Di Liberti}, Fosco Loregian, Chad Nester, and Pawel Sobocinski.
\newblock Functorial semantics for partial theories.
\newblock {\em Proc. {ACM} Program. Lang.}, 5({POPL}):1--28, 2021.

\bibitem[DPH87]{di1987dominical}
Robert~A Di~Paola and Alex Heller.
\newblock Dominical categories: recursion theory without elements1 2.
\newblock {\em The Journal of symbolic logic}, 52(3):594--635, 1987.

\bibitem[ELH15]{everitt15:sequential}
Tom Everitt, Jan Leike, and Marcus Hutter.
\newblock Sequential extensions of causal and evidential decision theory.
\newblock In {\em International Conference on Algorithmic Decision Theory}, pages 205--221. Springer, 2015.

\bibitem[EPT17]{ehrhard2017measurable}
Thomas Ehrhard, Michele Pagani, and Christine Tasson.
\newblock Measurable cones and stable, measurable functions: a model for probabilistic higher-order programming.
\newblock {\em Proceedings of the ACM on Programming Languages}, 2(POPL):1--28, 2017.

\bibitem[FGP21]{fitz2021deFinetti}
Tobias Fritz, Tomáš Gonda, and Paolo Perrone.
\newblock {De Finetti}’s theorem in categorical probability.
\newblock {\em Journal of Stochastic Analysis}, 2(4), 2021.

\bibitem[FL22]{fritz_liang_2022}
Tobias Fritz and Wendong Liang.
\newblock Free {GS}-monoidal categories and free {Markov} categories.
\newblock {\em CoRR}, abs/2204.02284, 2022.

\bibitem[Fon13]{fong_2013}
Brendan Fong.
\newblock Causal {Theories}: {A} {Categorical} {Perspective} on {Bayesian} {Networks}.
\newblock {\em Master's Thesis, University of Oxford. ArXiv preprint arXiv:1301.6201}, 2013.

\bibitem[Fox76]{fox76}
Thomas Fox.
\newblock Coalgebras and {Cartesian} {Categories}.
\newblock {\em Communications in Algebra}, 4(7):665--667, 1976.

\bibitem[FP19]{fritz2019probability}
Tobias Fritz and Paolo Perrone.
\newblock A probability monad as the colimit of spaces of finite samples.
\newblock {\em Theory and Applications of Categories}, 34(7):170--220, 2019.

\bibitem[FPR21]{fritz2021probability}
Tobias Fritz, Paolo Perrone, and Sharwin Rezagholi.
\newblock Probability, valuations, hyperspace: Three monads on top and the support as a morphism.
\newblock {\em Mathematical Structures in Computer Science}, 31(8):850--897, 2021.

\bibitem[FR20]{fritz2020infinite}
Tobias Fritz and Eigil~Fjeldgren Rischel.
\newblock Infinite products and zero-one laws in categorical probability.
\newblock {\em Compositionality}, 2:3, 2020.

\bibitem[Fri20]{fritz_2020}
Tobias Fritz.
\newblock A synthetic approach to {{Markov}} kernels, conditional independence and theorems on sufficient statistics.
\newblock {\em Advances in Mathematics}, 370:107239, 2020.

\bibitem[GH78]{gibbard_harper_1978}
Allan Gibbard and William~L. Harper.
\newblock Counterfactuals and two kinds of expected utility.
\newblock In {\em Ifs}, pages 153--190. Springer, 1978.

\bibitem[Gir82]{giry82:categorical}
Michèle Giry.
\newblock A categorical approach to probability theory.
\newblock In {\em Categorical aspects of topology and analysis}, pages 68--85. Springer, 1982.

\bibitem[Gre13]{greaves2013epistemic}
Hilary Greaves.
\newblock Epistemic decision theory.
\newblock {\em Mind}, 122(488):915--952, 2013.

\bibitem[GW06]{greaves2006justifying}
Hilary Greaves and David Wallace.
\newblock Justifying conditionalization: Conditionalization maximizes expected epistemic utility.
\newblock {\em Mind}, 115(459):607--632, 2006.

\bibitem[Hal17]{halpern2017reasoning}
Joseph~Y Halpern.
\newblock {\em Reasoning about uncertainty}.
\newblock MIT press, 2017.

\bibitem[Has97]{hasegawa97}
Masahito Hasegawa.
\newblock {\em Models of sharing graphs: a categorical semantics of let and letrec}.
\newblock PhD thesis, University of Edinburgh, {UK}, 1997.

\bibitem[HBH88]{horvitz88:decision}
Eric~J Horvitz, John~S Breese, and Max Henrion.
\newblock Decision theory in expert systems and artificial intelligence.
\newblock {\em International journal of approximate reasoning}, 2(3):247--302, 1988.

\bibitem[HHJW07]{haskell_2007}
Paul Hudak, John Hughes, Simon L.~Peyton Jones, and Philip Wadler.
\newblock A history of {Haskell}: being lazy with class.
\newblock In Barbara~G. Ryder and Brent Hailpern, editors, {\em Proceedings of the Third {ACM} {SIGPLAN} History of Programming Languages Conference (HOPL-III), San Diego, California, USA, 9-10 June 2007}, pages 1--55. {ACM}, 2007.

\bibitem[HKSY17]{heunen_kammar_staton_yang_2017}
Chris Heunen, Ohad Kammar, Sam Staton, and Hongseok Yang.
\newblock A convenient category for higher-order probability theory.
\newblock In {\em 2017 32nd Annual ACM/IEEE Symposium on Logic in Computer Science (LICS)}, pages 1--12. IEEE, 2017.

\bibitem[Jac18]{jacobs2018probability}
Bart Jacobs.
\newblock From probability monads to commutative effectuses.
\newblock {\em Journal of logical and algebraic methods in programming}, 94:200--237, 2018.

\bibitem[Jac19]{jacobs_2019}
Bart Jacobs.
\newblock The mathematics of changing one's mind, via jeffrey's or via pearl's update rule.
\newblock {\em J. Artif. Intell. Res.}, 65:783--806, 2019.

\bibitem[Jac21]{jacobs2021learning}
Bart Jacobs.
\newblock Learning from what's right and learning from what's wrong.
\newblock {\em arXiv preprint arXiv:2112.14045}, 2021.

\bibitem[Jef90]{jeffrey1990logic}
Richard~C Jeffrey.
\newblock {\em The logic of decision}.
\newblock University of Chicago press, 1990.

\bibitem[JKZ21]{jacobs2021causal}
Bart Jacobs, Aleks Kissinger, and Fabio Zanasi.
\newblock Causal inference via string diagram surgery: A diagrammatic approach to interventions and counterfactuals.
\newblock {\em Mathematical Structures in Computer Science}, 31(5):553--574, 2021.

\bibitem[Joy99]{joyce1999foundations}
James~M Joyce.
\newblock {\em The foundations of causal decision theory}.
\newblock Cambridge University Press, 1999.

\bibitem[JS91]{joyal91}
André Joyal and Ross Street.
\newblock The geometry of tensor calculus, i.
\newblock {\em Advances in Mathematics}, 88(1):55--112, 1991.

\bibitem[JZ20]{jacobs2020logical}
Bart Jacobs and Fabio Zanasi.
\newblock The logical essentials of bayesian reasoning.
\newblock {\em Foundations of Probabilistic Programming}, pages 295--331, 2020.

\bibitem[Lew81]{lewis1981causal}
David Lewis.
\newblock Causal decision theory.
\newblock {\em Australasian journal of philosophy}, 59(1):5--30, 1981.

\bibitem[ML71]{macLane1971}
Saunders Mac~Lane.
\newblock {\em Categories for the Working Mathematician}, volume~5 of {\em Graduate Texts in Mathematics}.
\newblock Springer Verlag, 1971.

\bibitem[Noz69]{nozick_1969}
Robert Nozick.
\newblock Newcomb’s {Problem} and {Two} {Principles} of {Choice}.
\newblock In {\em Essays in honor of Carl G. Hempel}, pages 114--146. Springer, 1969.

\bibitem[Pan99]{panangaden1999}
Prakash Panangaden.
\newblock The {{Category}} of {{Markov Kernels}}.
\newblock {\em Electronic Notes in Theoretical Computer Science}, 22:171--187, January 1999.

\bibitem[Pea88]{pearl1988probabilistic}
Judea Pearl.
\newblock {\em Probabilistic reasoning in intelligent systems: networks of plausible inference}.
\newblock Morgan kaufmann, 1988.

\bibitem[Pea90]{pearl1990jeffrey}
Judea Pearl.
\newblock Jeffrey’s rule, passage of experience, and neo-bayesianism.
\newblock In {\em Knowledge representation and defeasible reasoning}, pages 245--265. Springer, 1990.

\bibitem[Pea09]{pearl_2009}
Judea Pearl.
\newblock {\em Causality}.
\newblock Cambridge university press, 2009.

\bibitem[RR88]{robinson1988categories}
Edmund Robinson and Giuseppe Rosolini.
\newblock Categories of partial maps.
\newblock {\em Information and computation}, 79(2):95--130, 1988.

\bibitem[Sha81]{shafer1981jeffrey}
Glenn Shafer.
\newblock Jeffrey's rule of conditioning.
\newblock {\em Philosophy of Science}, 48(3):337--362, 1981.

\bibitem[SS21]{stein2021compositional}
Dario Stein and Sam Staton.
\newblock Compositional semantics for probabilistic programs with exact conditioning.
\newblock In {\em 2021 36th Annual ACM/IEEE Symposium on Logic in Computer Science (LICS)}, pages 1--13. IEEE, 2021.

\bibitem[Ste21]{stein_thesis_2021}
Dario~Maximilian Stein.
\newblock Structural foundations for probabilistic programming languages.
\newblock {\em University of Oxford}, 2021.

\bibitem[SV13]{stay2013bicategorical}
Mike Stay and Jamie Vicary.
\newblock Bicategorical semantics for nondeterministic computation.
\newblock {\em Electronic Notes in Theoretical Computer Science}, 298:367--382, 2013.

\bibitem[SWY{\etalchar{+}}16]{staton_et_al_2016}
Sam Staton, Frank Wood, Hongseok Yang, Chris Heunen, and Ohad Kammar.
\newblock Semantics for probabilistic programming: higher-order functions, continuous distributions, and soft constraints.
\newblock In {\em 2016 31st annual ACM/IEEE Symposium on Logic in Computer Science (LiCS)}, pages 1--10. IEEE, 2016.

\bibitem[VKS19]{vakar2019domain}
Matthijs V{\'a}k{\'a}r, Ohad Kammar, and Sam Staton.
\newblock A domain theory for statistical probabilistic programming.
\newblock {\em Proceedings of the ACM on Programming Languages}, 3(POPL):1--29, 2019.

\bibitem[vS]{vosSavant}
Marilyn vos Savant.
\newblock {Parade 16: Ask Marilyn (Archived)}.
\newblock \url{https://web.archive.org/web/20130121183432/http://marilynvossavant.com/game-show-problem/}.
\newblock Accessed: 2013-01-21.

\bibitem[YS17]{yudkowsky_soares_2017}
Eliezer Yudkowsky and Nate Soares.
\newblock Functional {Decision} {Theory}: a {New} {Theory} of {Instrumental} {Rationality}.
\newblock {\em ArXiv preprint arXiv:1710.05060}, 2017.

\end{thebibliography}
\newpage
\appendices
\section{Preliminaries}\label{sec:app-detailed-background}
A symmetric monoidal category $(\catC, \otimes, I)$ is said to be \emph{cartesian monoidal} whenever the tensor $A \tensor B$ of two objects $A$ and $B$ is their categorical product, and associators and unitors are derived from the universal property of the categorical product.

A well-known characterization of cartesian monoidal categories is \emph{Fox's theorem}~\cite{fox76}, which states that a symmetric monoidal category is cartesian if and only if it is a copy-discard category where the comonoid structure is natural.
We present a similar characterization: cartesian monoidal categories are monoidal categories where all joint maps split.

\begin{proposition}
  A symmetric monoidal category \((\cat{C}, \tensor, \monoidalunit)\) is cartesian if and only if
  (i) every object \(X \in \obj{\cat{C}}\) has a uniform comonoid structure \((X,\cp,\discard)\), meaning that \(\discard_{X \tensor Y} = \discard \tensor \discard_Y\), \(\discard_{\monoidalunit} = \id{\monoidalunit}\) and that \(\cp_{X \tensor Y} = (\cp_X \tensor \cp_Y) \dcomp (\id{X} \tensor \swap{X,Y} \tensor \id{Y})\), \(\cp_{\monoidalunit} = \id{\monoidalunit}\);
  (ii) every morphism \(f \colon X \to \monoidalunit\) is equal to \(\discard_{X}\);
  and (iii) every morphism \(f \colon X \to Y_1 \tensor Y_2\) splits as \(\cp_X \dcomp (f_1 \tensor f_2)\) for some \(f_1 \colon X \to Y_1\) and \(f_2 \colon X \to Y_2\).
  \begin{figure}[h!]
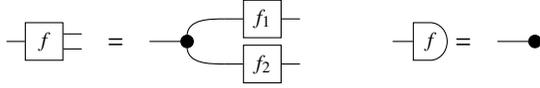

    \cartesiancatfactorisationFig{}
    \caption{Cartesian category}\label{diagram-cartesian-split}
  \end{figure}
\end{proposition}
\begin{proof}
  We will show that the comonoid structure is natural.
  Firstly, we know from the premises that the discard is natural: that is, \(f \dcomp \discard_Y = \discard_X\) for each \(f \colon X \to Y\).
  We now show that every time we split \(f  \colon X \to Y_1 \tensor Y_2\) into \(f_1\) and \(f_2\), we are forced to admit that \(f_1 = f \dcomp (\id{Y_1} \tensor \discard_{Y_2})\) and that \(f_2 = f \dcomp (\discard_{Y_1} \tensor \id{Y_2})\).
  This implies in turn that copying must be a natural transformation: given any \(g \colon X \to Y\), we know that \(f = g \dcomp \cp_Y\) must split into \(f = \cp_X \dcomp (g \tensor g)\).
  By Fox's theorem, the category is cartesian.
\end{proof}

We spell out the definitions of Markov and partial Markov category in detail.

\begin{definition}
    A \emph{Markov category} is a symmetric monoidal category $(\catC, \tensor, I)$ such that (i) every object $X \in \obj{\catC}$ has a uniform comonoid structure $(X,\cp_X,\discard_X)$, meaning that $\discard_{X \tensor Y} = \discard_X \tensor \discard_Y$, $\discard_I = \im$ and that \(\cp_{X \tensor Y} = (\cp_X \tensor \cp_Y) \comp (\id{X} \tensor \swap[X,Y] \tensor \id{Y})\), $\cp_I = \im$; (ii) every morphism $f \colon X \to I$ is equal to $\discard_X$; and (iii) every morphism $f \colon X \to Y_1 \tensor Y_2$ splits both as \(\cp \comp (\id{} \tensor (f_{2} \comp \cp_{Y_{2}})) \comp (g_1 \tensor \id{})\) and as \(\cp \comp ((f_{1} \comp \cp_{Y_{1}}) \tensor \id{}) \comp (\id{} \tensor g_2)\) for some $f_1 \colon X \to Y_1$, some $f_2 \colon X \to Y_2$, some $g_1 \colon X \tensor Y_2 \to Y_1$ and some $g_2 \colon Y_{1} \tensor X \to Y_2$.

\end{definition}

\begin{remark}[Conditionals in \(\kleisli{\distr}\)]\label{rem:distr-conditionals}
  Given a morphism \(f \colon X \to A \tensor B\), its \emph{marginal} on \(A\) is
  \[m(a \given x) \defn \sum_{b \in B} f(a,b \given x),\]
  and a \emph{conditional} with respect to \(A\) is
  \[c(b \given a,x) \defn \begin{cases} \dfrac{f(a,b \given x)}{m(a \given x)}, & \mbox{ when } m(a \given x) \neq 0; \\ \sigma(b), & \mbox{ when } m(a \given x) = 0. \end{cases}\]
  for any distribution \(\sigma \colon \monoidalunit \to B\).
\end{remark}

\begin{proof}[\textbf{Proof of \Cref{prop:finitary-subdistributions}}]
  For any \(f \colon X \to A \tensor B\), let \(m \colon X \to A\) indicate \(m \defn f \dcomp (\id{A} \tensor \discard_{B})\).
  Explicitly, this means that \(m(a \given x) = \sum_{b \in B} f(a,b \given x)\) and \(m(\bot \given x) = f(\bot \given x)\).
  We define a map \(c \colon A \tensor X \to B\) and show that it is a conditional of \(f\) with marginal \(m\).
  \begin{align*}
    c(b \given a,x) & \defn
      \begin{cases}
        \dfrac{f(a,b \given x)}{m(a \given x)},
          & \mbox{ if } m(a \given x) \neq 0; \\
        0, & \text{otherwise};
      \end{cases}\\
    c(\bot \given a,x) &  \defn
      \begin{cases}
        0, & \mbox{ if } m(a \given x) \neq 0; \\
        1, & \text{otherwise}.
      \end{cases}
  \end{align*}
  This definition gives a map in \(\kleisli{\subdistr}\) because, for all \(a \in A\) and \(x \in X\) such that \(m(a \given x) \neq 0\), the total probability mass is \(\sum_{b \in B} c(b \given a,x) + c(\bot \given a,x) = \sum_{b \in B} \dfrac{f(a,b \given x)}{m(a \given x)} + 0 = 1\), and, for all \(a \in A\) and \(x \in X\) such that \(m(a \given x) = 0\), the total probability mass is \(\sum_{b \in B} c(b \given a,x) + c(\bot \given a,x) = 0 + 1 = 1\).

  We show that \(m \condcomp c = f\) (as defined in \Cref{rem:conditional-composition} and \Cref{fig:conditionals}).
  \begin{align*}
    m \condcomp c (a,b \given x)
      & = m(a \given x) \cdot c(b \given a,x)\\
      & = \begin{cases}
        m(a \given x) \cdot \dfrac{f(a,b \given x)}{m(a \given x)}, &
        \mbox{ if } m(a \given x) \neq 0; \\
        0, & \mbox{ if } m(a \given x) = 0;
      \end{cases} \\
      & = \begin{cases}
        f(a,b \given x), & \mbox{ if }
        \exists b' \in B. \ f(a,b' \given x) \neq 0; \\
        0, & \forall b' \in B. \ f(a,b' \given x) = 0;
      \end{cases} \\
      & = f(a,b \given x). \\
      m \condcomp c (\bot \given x) & =
      m(\bot \given x) + \sum_{a \in A} m(a \given x) c(\bot \given a,x) \\
                                 & = f(\bot \given x) + 0 \\
                                 & = f(\bot \given x),
  \end{align*}
  which proves these are a marginal and a conditional.
\end{proof}

\begin{definition}\label{def:measurable-spaces}
  A \emph{measurable space} \((X, \sigma_{X})\) is a set \(X\) equipped with a \(\sigma\)-algebra \(\sigma_{X}\) on \(X\).
  A \emph{measurable function} \(f \colon (X,\sigma_{X}) \to (Y,\sigma_{Y})\) is given by a function \(f \colon X \to Y\) such that, for each set \(E \in \sigma_{Y}\), its preimage under \(f\) belongs to the \(\sigma\)-algebra on \(X\): \(f^{-1}(E) \in \sigma_{X}\).
  Measurable spaces and measurable functions form a cartesian monoidal \defining{linkmeas}{category} \(\Meas\) with composition lifted from the category \(\Set\) of sets and functions and cartesian product given by \((X,\sigma_{X}) \times (Y,\sigma_{Y}) \defn (X \times Y, \sigma_{X} \times \sigma_{Y})\), where \(\sigma_{X} \times \sigma_{Y}\) is the \(\sigma\)-algebra generated by subsets of the form \(E \times F\) with \(E \in \sigma_{X}\) and \(F \in \sigma_{Y}\).
\end{definition}

\begin{definition}\label{def:standard-borel-spaces}
  A measurable space \((X,\sigma_{X})\) is a standard Borel space if there is a metric on \(X\) that makes \(X\) a complete separable metric space and \(\sigma_{X}\) the Borel \(\sigma\)-algebra on \(X\).
  We indicate with \(\sigmaBorel{X}\) the Borel \(\sigma\)-algebra on a set \(X\) when the metric can be deduced from the context.
\end{definition}

\begin{definition}\label{def:probability-measure}
  A \emph{measure} \(p\) on a measurable space \((X, \sigma_{X})\) is given by a function \(p \colon \sigma_{X} \to \reals^{+}_{\infty}\) to the extended positive reals such that \(p(\emptyset) = 0\) and, for each sequence of disjoint sets \(X_{i}\) with \(i \in \naturals\), \(p(\Union_{i \in \naturals} X_{i}) = \sum_{i \in \naturals} p(X_{i})\).
  A measure \(p\) is a \emph{probability measure} if, additionally, its total mass is \(1\): \(p(X) = 1\).
  We indicate the set of probability measures on a measurable space \((X, \sigma_{X})\), endowed with the \(\sigma\)-algebra generated by the sets \(S_{T,E} \defn \{p \st p(E) \in T\}\) for each \(T \in \sigmaBorel{[0,1]}\) and \(E \in \sigma_{X}\), by \(\Giry(X, \sigma_{X})\).
  When the \(\sigma\)-algebra \(\sigma_{X}\) can be deduced from the context, we simply write \(\Giry(X)\).
\end{definition}

\begin{definition}
  The category \(\Gauss\) models affine maps with gaussian-distributed noise. It has natural numbers as objects; and
  a morphism \(f \colon n \to m\) is given by a tuple \((M,C,s)\) of matrices over the reals, with \(M \in \Mat{\reals}(m,n)\), \(C \in \Mat{\reals}(m,m)\) and \(s \in \Mat{\reals}(m,1)\).
  This data defines the conditional distribution of a random variable \(Y\) in \(\reals^{m}\) as an affine transformation of a random variable \(X\) in \(\reals^{n}\): \(Y = M \cdot X + S\), where \(S\) is a Gaussian random variable independent of \(X\) with mean \(s\) and covariance matrix \(C\).
  The composition of two morphisms \((N,D,t) \colon p \to n\) and \((M,C,s) \colon n \to m\) is a morphism \((M \cdot N, M \cdot D \cdot \transpose{M} + C, M \cdot t + s)\), while their monoidal product is the morphism \((N \biproduct M, D \biproduct C, \binom{t}{s})\).
  The identity morphism is \((\idmat{n},\zeromat{n},0)\).
\end{definition}

\section{Partial Markov categories}\label{sec:app-partial-markov}

\begin{definition}
  A \emph{partial Markov category} is a symmetric monoidal category $(\catC, \tensor, I)$ such that
  (i) every object $X \in \obj{\catC}$ has a partial Frobenius monoid (\Cref{diagram-partial-frobenius}) structure $(X,\cp_X,\discard_X,\mu_X)$ which is uniform, meaning that $\discard_{X \tensor Y} = \discard_X \tensor \discard_Y$, $\discard_I = \im$, \(\cp_{X \tensor Y} = (\cp_X \tensor \cp_Y) \comp (\id{X} \tensor \swap[X,Y] \tensor \id{Y})\), $\cp_I = \im$, \(\compare_{X \tensor Y} = (\id{X} \tensor \swap[X,Y] \tensor \id{Y}) \comp (\compare_{X} \tensor \compare_{Y})\), and \(\compare_{\monoidalunit} = \id{}\);
  and (ii) every morphism $f \colon X \to Y_1 \tensor Y_2$ splits both as \(\cp \comp (\id{} \tensor (f_{2} \comp \cp_{Y_{2}})) \comp (g_1 \tensor \id{})\) and as \(\cp \comp ((f_{1} \comp \cp_{Y_{1}}) \tensor \id{}) \comp (\id{} \tensor g_2)\) for some $f_1 \colon X \to Y_1$, some $f_2 \colon X \to Y_2$, some $g_1 \colon X \tensor Y_2 \to Y_1$ and some $g_2 \colon Y_{1} \tensor X \to Y_2$.
\end{definition}

\begin{proof}[Proof of \Cref{prop:characterisation-quasi-total}]
  Suppose \(f \colon X \to Y\) is quasi-total.
  By composing the both sides of the equation that defines quasi-totality with the discard map \(\discard\), we obtain that the probability of success of \(f\) is deterministic.
  \begin{align*}
    \quasitotalThenDeterministicFig{}
  \end{align*}

  Conversely, suppose that \(f\) has a domain of definition, i.e. that its probability of success \(f \dcomp \discard\) is deterministic.
  \begin{align}
    \deterministicThenQuasitotalFig{}
  \end{align}
  \Cref{eq:quasitotal-unitality1,eq:quasitotal-unitality2,eq:quasitotal-associativity} follow from counitality and coassociativity of the copy-discard structure, \Cref{eq:quasitotal-conditional1,eq:quasitotal-conditional2} are applications of quasi-total conditionals and \Cref{prop:QTconditionals-give-nice-marginals}, and \Cref{eq:quasitotal-determinism} follows from the hypothesis of determinism of \(f \dcomp \discard\).
\end{proof}

\section{Proof of distributive law}\label{sec:distributive law}
\begin{proposition}\label{prop:composing-monads}
  Let \(\monad, \monad_{1},\monad_{2} \colon \cat{C} \to \cat{C}\) be a monads with multiplications \(\monadmultiplication\), \(\monadmultiplication_{1}\), and \(\monadmultiplication_{2}\), and units \(\monadunit\), \(\monadunit_{1}\) and \(\monadunit_{2}\), respectively.
  Suppose that:
  \begin{enumerate}
    \item the underlying functor of \(\monad\) is the composition of the underlying functors of \(\monad_{1}\) and \(\monad_{2}\): \(\monad = \monad_{1} \dcomp \monad_{2}\);
    \item there is a natural transformation \(\distributivelaw \colon \monad_{2} \dcomp \monad_{1} \to \monad_{1} \dcomp \monad_{2}\);
    \item \(\monadmultiplication = (\id{} \tensor \distributivelaw \tensor \id{}) \dcomp (\monadmultiplication_{1} \tensor \monadmultiplication_{2})\);
    \item \(\monadunit = \monadunit_{1} \tensor \monadunit_{2}\).
  \end{enumerate}
  Then, the transformation \(\distributivelaw\) is a distributive law of monads.
\end{proposition}
\begin{proof}
  The proof is a string diagrammatic reasoning in the monoidal category of endofunctors on \(\cat{C}\) and natural transformations between them.

  \Cref{eq:right-unitality-composite-monoid} represents right unitality of \(\monadmultiplication\).
  \begin{equation}\label{eq:right-unitality-composite-monoid}
    \rightunitalitycompositemonoidFig{}
  \end{equation}
  By precomposing \Cref{eq:right-unitality-composite-monoid} with \((\monadunit_{1} \tensor \id{\monad_{2}})\), we obtain \Cref{eq:right-unit-unitality1}.
  Unitality of \(\monadmultiplication_{1}\) and \(\monadmultiplication_{2}\) imply (\ref{eq:right-unit-unitality2}).
  \distributivelawrightunitaxiomproofFig{}
  This shows that the unit \(\monadunit_{1}\) commutes with \(\distributivelaw\).
  Similarly, by applying left unitality of \(\monadmultiplication\), one can show that the unit \(\monadunit_{2}\) commutes with \(\distributivelaw\).
  \Cref{eq:associativity-composite-monoid} represents associativity of \(\monadmultiplication\).
  \begin{equation}\label{eq:associativity-composite-monoid}
    \associativitycompositemonoidFig{}
  \end{equation}
  By precomposing \Cref{eq:associativity-composite-monoid} with \((\monadunit_{1} \tensor \id{\monad_{2}} \tensor \monadunit_{1} \tensor \id{\monad_{2}} \tensor \id{\monad_{1}} \tensor \monadunit_{2})\), we obtain (\ref{eq:left-multiplication-associativity}).
  Equations (\ref{eq:left-multiplication-unitality1}) and (\ref{eq:left-multiplication-unitality2}) follow from unitality of \(\monadmultiplication_{1}\) and \(\monadmultiplication_{2}\).
  Lastly, we showed above that the unit \(\monadunit_{1}\) commutes with \(\distributivelaw\).
  This implies \Cref{eq:left-multiplication-rightunit1}.
  \distributivelawleftmultiplicationaxiomproofFig{}
  This shows that the multiplication \(\monadmultiplication_{2}\) commutes with the distributive law.
  By a similar reasoning, the multiplication \(\monadmultiplication_{1}\) commutes with the distributive law.
\end{proof}

\begin{lemma}\label{lemma:distributive-law-natural}
  The components \(\distributivelaw_{X}\) defined in \Cref{rem:subBorelStoch-kleisli-cat} form a natural transformation \(\distributivelaw \colon \maybe[\Giry(-)] \to \Giry(\maybe)\).
\end{lemma}
\begin{proof}
  We define the components of \(\distributivelaw\) as \(\distributivelaw_{X}(\sigma) \defn \extenddomain{\sigma}\) and \(\distributivelaw_{X}(\bot) \defn \dirac{\bot}\), where \(\extenddomain{\sigma}\) is the extension of \(\sigma\) to \(\maybe[X]\) by \(\extenddomain{\sigma}(\{\bot\}) \defn 0\).
  Let \(f \colon (X,\sigma_{X}) \to (Y,\sigma_{Y})\).
  We show that the naturality square commutes for \(p \in \Giry(X)\):
  \begin{align*}
    &p \quad \overset{\maybe[\Giry(f)]}{\mapsto} \quad f^{-1}\dcomp p \quad \overset{\distributivelaw_{Y}}{\mapsto} \quad \extenddomain{(f^{-1}\dcomp p)}\\
    &p \quad \overset{\distributivelaw_{x}}{\mapsto} \quad \extenddomain{p} \quad \overset{\maybe[\Giry(f)]}{\mapsto} \quad (f+1)^{-1} \dcomp \extenddomain{p}
  \end{align*}
  The distributions \((f+1)^{-1} \dcomp \extenddomain{p}\) and \(\extenddomain{(f^{-1}\dcomp p)}\) coincide because they have the same values on measurable sets \(T \in \sigma_{Y}\).
  For the distinguished element \(\bot\), the naturality square also commutes:
  \begin{align*}
    &\bot \quad \overset{\maybe[\Giry(f)]}{\mapsto} \quad \bot \quad \overset{\distributivelaw_{Y}}{\mapsto} \quad \dirac{\bot}\\
    &\bot \quad \overset{\distributivelaw_{x}}{\mapsto} \quad \dirac{\bot} \quad \overset{\maybe[\Giry(f)]}{\mapsto} \quad \dirac{\bot}
  \end{align*}
\end{proof}

\section{Proofs from \Cref{sec:partial-markov}}\label{sec:app-proofs-partial-markov}

\begin{proof}[\textbf{Proof of~\Cref{prop:partial-markov-from-markov}}]
  The Kleisli category \((\kleisli{\monad}, \ktensor, \monoidalunit)\) is monoidal because \(\monad\) is monoidal.
  We check that \(\kleisli{\monad}\) has conditionals.
  We indicate with \(f \colon A \kto B\) morphisms in \(\kleisli{\monad}\) and with \(f \kcomp g\) their composition.
  Let \(f \colon X \kto A \ktensor B\).
  This means that \(f \colon X \to \monad(A \tensor B)\) and \(f \dcomp \laxatorsec_{A,B} \colon X \to \monad A \tensor \monad B\).
  Since \(\cat{C}\) has conditionals, there are a marginal \(m \colon A \to \monad A\) and a conditional \(c \colon \monad A \tensor X \to \monad B\) such that \(f \dcomp \laxatorsec_{A,B} = m \condcomp c\), where \(m \condcomp c\) is defined in \Cref{rem:conditional-composition} and \Cref{fig:conditionals}.
  We want to find a marginal \(m' \colon X \kto A\) and a conditional \(c' \colon A \ktensor X \kto B\) such that \(m' \kcondcomp c' = f\), where \(\kcondcomp\) indicates the operation \(\condcomp\) instantiated in \(\kleisli{\monad}\).
  Good candidates for the marginal and the conditional are \(m' \defn m\) and \(c' \defn (\monadunit_A \tensor \id{X}) \dcomp c\).
  We check that the desired equation holds.
  \begin{align*}
      & m \kcondcomp c' \\
      \explain[\defn]{by definition of \(\kcondcomp\)}\\
      & \cp_X \kcomp (m' \ktensor \id{X}) \kcomp (\cp_{A} \ktensor \id{X}) \kcomp (\id{A} \ktensor c') \\
      \explain{by definition of \(\kcomp\) and \(\ktensor\)}\\
      & \cp_X \dcomp \monadunit_{X \tensor X} \dcomp \monad((m \tensor \monadunit_X) \dcomp \laxator_{A,X}) \dcomp \monadmultiplication_{A \tensor X} \\
      & \qquad \dcomp \monad(((\cp_A \dcomp \monadunit_{A \tensor A}) \tensor \monadunit_X) \dcomp \laxator_{A \tensor A, X}) \dcomp \monadmultiplication_{A \tensor A \tensor X}\\
      & \qquad \dcomp \monad((\monadunit_A \tensor c') \dcomp \laxator_{A,B}) \dcomp \monadmultiplication_{A \tensor B}\\
      \explain{by naturality and monoidality of \(\monadunit\)}\\
      & \cp_X \dcomp (m \tensor \monadunit_X) \dcomp \laxator_{A,X} \dcomp \monadunit_{\monad(A \tensor X)} \dcomp \monadmultiplication_{A \tensor X} \dcomp \monad(\cp_A \tensor \id{X}) \\
      & \qquad \dcomp \monad\monadunit_{A \tensor A \tensor X} \dcomp \monadmultiplication_{A \tensor A \tensor X} \dcomp \monad(\id{A} \tensor c') \\
      & \qquad \dcomp \monad(\monadunit_A \tensor \monad\id{B}) \dcomp \monad\laxator_{A,B} \dcomp \monadmultiplication_{A \tensor B}\\
      \explain{by unitality of \(\monadmultiplication\)}\\
      & \cp_X \dcomp (m \tensor \monadunit_X) \dcomp \laxator_{A,X} \dcomp \monad(\cp_A \tensor \id{X}) \dcomp \monad(\id{A} \tensor c')\\
      & \qquad \dcomp \monad(\monadunit_A \tensor \monad\id{B}) \dcomp \monad\laxator_{A,B} \dcomp \monadmultiplication_{A \tensor B}\\
      \explain{because \(\laxator_{A,B}\) is a split epimorphism}\\
      & \cp_X \dcomp (m \tensor \monadunit_X) \dcomp \laxator_{A,X} \dcomp \monad(\cp_A \tensor \id{X}) \dcomp \monad(\id{A} \tensor c') \\
      & \qquad \dcomp \monad(\monadunit_A \tensor \monad\id{B}) \dcomp \laxatorsec_{\monad A, \monad B} \dcomp \laxator_{\monad A, \monad B} \dcomp \monad\laxator_{A,B} \dcomp \monadmultiplication_{A \tensor B}\\
      \explain{by monoidality of \(\monadmultiplication\)}\\
      & \cp_X \dcomp (m \tensor \monadunit_X) \dcomp \laxator_{A,X} \dcomp \monad(\cp_A \tensor \id{X}) \dcomp \monad(\id{A} \tensor c') \\
      & \qquad \dcomp \monad(\monadunit_A \tensor \monad\id{B}) \dcomp \laxatorsec_{\monad A, \monad B} \dcomp (\monadmultiplication_A \tensor \monadmultiplication_B) \dcomp \laxator_{A, B}\\
      \explain{by naturality of \(\laxatorsec\)}\\
      & \cp_X \dcomp (m \tensor \monadunit_X) \dcomp \laxator_{A,X} \dcomp \monad(\cp_A \tensor \id{X}) \dcomp \laxatorsec_{A, A \tensor X} \\
      & \qquad \dcomp (\monad\id{A} \tensor \monad c') \dcomp ((\monad\monadunit_A \dcomp \monadmultiplication_A) \tensor \monadmultiplication_B) \dcomp \laxator_{A, B}\\
      \explain{by unitality of \(\monadmultiplication\)}\\
      & \cp_X \dcomp (m \tensor \monadunit_X) \dcomp \laxator_{A,X} \dcomp \monad(\cp_A \tensor \id{X}) \dcomp \laxatorsec_{A, A \tensor X} \\
      & \qquad \dcomp (\monad\id{A} \tensor \monad c') \dcomp (\monad\id{A} \tensor \monadmultiplication_B) \dcomp \laxator_{A, B}\\
      \explain{by definition of \(c'\)}\\
      & \cp_X \dcomp (m \tensor \monadunit_X) \dcomp \laxator_{A,X} \dcomp \monad(\cp_A \tensor \id{X}) \dcomp \laxatorsec_{A, A \tensor X} \\
      & \qquad \dcomp (\monad\id{A} \tensor \monad \monadunit_A) \dcomp (\monad\id{A} \tensor \monad c) \dcomp (\monad\id{A} \tensor \monadmultiplication_B) \dcomp \laxator_{A, B}\\
      \explain{by naturality of \(\laxator\) and \(\monadunit\)}\\
      & \cp_X \dcomp (m \tensor \id{X}) \dcomp ((\monad\cp_A \dcomp \laxatorsec_{A,A} \dcomp (\monad\id{A} \tensor \monad \monadunit_A)) \tensor \id{X}) \\
      & \qquad \dcomp (\monad\id{A} \tensor \monad\monad\id{A} \tensor \monadunit_X) \dcomp (\monad\id{A} \tensor \laxator_{\monad A, X}) \\
      & \qquad \dcomp (\monad\id{A} \tensor \monad c) \dcomp (\monad\id{A} \tensor \monadmultiplication_B) \dcomp \laxator_{A, B}\\
      \explain{by the first assumption on \(\cp\)}\\
      & \cp_X \dcomp (m \tensor \id{X}) \dcomp ((\cp_{\monad A} \dcomp (\monad\id{A} \tensor \monad \monadunit_A)) \tensor \id{X}) \\
      & \qquad \dcomp (\monad\id{A} \tensor \monad\monad\id{A} \tensor \monadunit_X) \dcomp (\monad\id{A} \tensor \laxator_{\monad A, X}) \\
      & \qquad \dcomp (\monad\id{A} \tensor \monad c) \dcomp (\monad\id{A} \tensor \monadmultiplication_B) \dcomp \laxator_{A, B}\\
      \explain{by the second assumption on \(\cp\)}\\
      & \cp_X \dcomp (m \tensor \id{X}) \dcomp ((\cp_{\monad A} \dcomp (\monad\id{A} \tensor \monadunit_{\monad A})) \tensor \id{X}) \\
      & \qquad \dcomp (\monad\id{A} \tensor \monad\monad\id{A} \tensor \monadunit_X) \dcomp (\monad\id{A} \tensor \laxator_{\monad A, X}) \\
      & \qquad \dcomp (\monad\id{A} \tensor \monad c) \dcomp (\monad\id{A} \tensor \monadmultiplication_B) \dcomp \laxator_{A, B}\\
      \explain{by naturality and monoidality of \(\monadunit\)}\\
      & \cp_X \dcomp (m \tensor \id{X}) \dcomp (\cp_{\monad A} \tensor \id{X}) \\
      & \qquad \dcomp (\monad\id{A} \tensor (\monadunit_{\monad A \tensor X} \dcomp \monad c \dcomp \monadmultiplication_B)) \dcomp \laxator_{A, B}\\
      \explain{by naturality of \(\monadunit\)}\\
      & \cp_X \dcomp (m \tensor \id{X}) \dcomp (\cp_{\monad A} \tensor \id{X}) \\
      & \qquad \dcomp (\monad\id{A} \tensor (c \dcomp \monadunit_{\monad B} \dcomp \monadmultiplication_B)) \dcomp \laxator_{A, B}\\
      \explain{by unitality of \(\monadmultiplication\)}\\
      & \cp_X \dcomp (m \tensor \id{X}) \dcomp (\cp_{\monad A} \tensor \id{X}) \dcomp (\monad\id{A} \tensor c) \dcomp \laxator_{A, B}\\
      \explain{by definition of \(\condcomp\)}\\
      & (m \condcomp c) \dcomp \laxator_{A,B}
  \end{align*}

  Then, \(m \kcondcomp c' = (m \condcomp c) \dcomp \laxator_{A,B} = f \dcomp \laxatorsec_{A,B} \dcomp \laxator_{A,B} = f\).
\end{proof}

\begin{proof}[\textbf{Proof of \Cref{lemma:quasi-total-conditionals}}]
  By spelling out the definition of quasi-total morphism (\Cref{def:quasi-total-morphism}) in \(\kleisli{\monad}\), we show that it becomes the condition in the statement.
  \begin{align*}
    & \cp \dcomp \monadunit \dcomp \monad(\monadunit \tensor \id{} \tensor \laxatorsec \tensor \monadunit \tensor \id{}) \dcomp (\monad c \tensor \monad c)\\
    & \qquad \dcomp (\monadmultiplication \tensor \monadmultiplication) \dcomp (\monad \id{} \tensor (\monad \discard \dcomp \monad \monadunit \dcomp \monadmultiplication)) \dcomp \laxator\\
    \explain{by unitality of \(\monadmultiplication\)}\\
    & \cp \dcomp \monadunit \dcomp \monad(\monadunit \tensor \id{} \tensor \laxatorsec \tensor \monadunit \tensor \id{}) \dcomp (\monad c \tensor \monad c)\\
    & \qquad \dcomp (\monadmultiplication \tensor \monadmultiplication) \dcomp (\monad \id{} \tensor \monad \discard) \dcomp \laxator\\
    \explain{by naturality of \(\monadmultiplication\)}\\
    & \cp \dcomp \monadunit \dcomp \monad(\monadunit \tensor \id{} \tensor \laxatorsec \tensor \monadunit \tensor \id{}) \dcomp (\monad c \tensor \monad c)\\
    & \qquad \dcomp (\monad\monad \id{} \tensor \monad\monad \discard) \dcomp (\monadmultiplication \tensor \monadmultiplication) \dcomp \laxator\\
    \explain{by monoidality of \(\monadmultiplication\)}\\
    & \cp \dcomp \monadunit \dcomp \monad(\monadunit \tensor \id{} \tensor \laxatorsec \tensor \monadunit \tensor \id{}) \dcomp (\monad c \tensor \monad c)\\
    & \qquad \dcomp (\monad\monad \id{} \tensor \monad\monad \discard) \dcomp \laxator \dcomp \laxator \dcomp \monadmultiplication\\
    \explain{by naturality of \(\monadunit\) and \(\laxatorsec \dcomp \laxator = \id{}\)}\\
    & \cp \dcomp (\monadunit \tensor \id{} \tensor \monadunit \tensor \id{}) \dcomp (c \tensor c) \dcomp (\monad \id{} \tensor \monad \discard) \dcomp \laxator \dcomp \monadunit \dcomp \monadmultiplication\\
    \explain{by unitality of \(\monadmultiplication\)}\\
    & \cp \dcomp (\monadunit \tensor \id{} \tensor \monadunit \tensor \id{}) \dcomp (c \tensor c) \dcomp (\monad \id{} \tensor \monad \discard) \dcomp \laxator\\
    \explain{by determinism of \(\monadunit\)}\\
    & (\monadunit \tensor \id{}) \dcomp \cp \dcomp (c \tensor c) \dcomp (\monad \id{} \tensor \monad \discard) \dcomp \laxator\\
    \explain{by assumption}\\
    & (\monadunit \tensor \id{}) \dcomp c
  \end{align*}
\end{proof}

\begin{proof}[\textbf{Proof of \Cref{lemma:borel-quasi-total}}]
  Let \(f \colon X \to A \tensor (\maybe[B])\).
  On some pairs \((a,x)\), the conditional \(c\) is determined by \(f\) and its marginal.
  On the pairs \((a,x)\) in which it is not, we can set \(c(a,x) \defn \dirac{\bot}\).
  This makes it satisfy the condition in \Cref{lemma:quasi-total-conditionals}.
\end{proof}

\begin{proof}[\textbf{Proof of \Cref{lemma:maybe-monad-satisfies-assumptions}}]
  Recall that the structural transformation \(\laxator\) for the Maybe monad is defined by \(\laxator_{A,B} \defn \id{A \tensor B} + \coproductmap{\finmap{A + B}}{\id{1}}\).
  We define a candidate for the section \(\laxatorsec\) of the structural transformation \(\laxator\) as \(\laxatorsec_{A,B} \defn \id{A \tensor B} + \initmap{A+B} + \id{1}\).
  Note that \(\laxatorsec\) is a natural transformation because it is a coproduct of natural transformations.
  With this definition, \(\laxatorsec\) is a section of \(\laxator\):
  \begin{align*}
    &  && \laxatorsec_{A,B} \dcomp \laxator_{A,B}\\
    & = && (\id{A \tensor B} + \initmap{A+B} + \id{1}) \dcomp (\id{A \tensor B} + \coproductmap{\finmap{A + B}}{\id{1}})\\
    & = && \id{A \tensor B} + \coproductmap{\initmap{A+B} \dcomp \finmap{A+B}}{\id{1}}\\
    & = && \id{A \tensor B} + \coproductmap{\initmap{1}}{\id{1}}\\
    & = && \id{A \tensor B} + \id{1}\\
    & = && \id{(A \tensor B) +1}
  \end{align*}

  For the second assumption, we need to show that \((\cp_{A} +1) \dcomp \laxatorsec_{A,A}\) coincides with the copy map \(\cp_{A+1}\).
  By \Cref{prop:markov-comonoid-is-copy}, it suffices to realize that \((\cp_{A} +1) \dcomp \laxatorsec_{A,A}\) forms a comonoid with unit \(\discard_{A+1}\).
  Associativity follows from associativity of \(\cp_{A}\) and naturality of \(\laxatorsec\).
  Unitality follows, less trivially, from unitality of \(\cp_{A}\) and naturality of \(\laxatorsec\).
  \begin{align*}
    &&& (\cp_{A} + 1) \dcomp \laxatorsec_{A,A} \dcomp (\id{A+1} \tensor \discard_{A+1})\\
    &= && (\cp_{A} + 1) \dcomp \laxatorsec_{A,A} \dcomp ((\id{A}+1) \tensor (\discard_{A}+1)) \dcomp (\id{A+1} \tensor \cocopy_{1})\\
    &= && ((\cp_{A} \dcomp (\id{A} \tensor \discard_{A}))+1) \dcomp \laxatorsec_{A,A} \dcomp (\id{A+1} \tensor \cocopy_{1})\\
    &= && (\id{A}+1) \dcomp \laxatorsec_{A,A} \dcomp \cocopy_{A+1}\\
    &= && (\id{A}+1) \dcomp \id{A+1}\\
    &= && \id{A+1}
  \end{align*}
  Similarly, one can prove that \((\cp_{A} + 1) \dcomp \laxatorsec_{A,A} \dcomp (\discard_{A+1} \tensor \id{A+1}) = \id{A+1}\).

  For the third assumption, we rewrite the two sides of the equation in normal form, using distributivity of \(\tensor\) over \(+\), and unitality of the universal map from the coproduct \(\cocopy_{A} \colon A + A \to A\).
  The left side becomes:
  \begin{align*}
    &\ ((\cp_{A+1} \dcomp (\id{A+1} \tensor (\monadunit_{A}+1))) \tensor \id{X}) \dcomp (\id{A+1} \tensor f) \dcomp \laxator_{A,B} \\
    =& \ ((\cp_{A} \tensor \id{X}) + \initmap{A \tensor X} + \initmap{A \tensor X} + \initmap{A \tensor X} + \initmap{1} + \id{X}) \dcomp ((\id{A} \tensor f) + f) \\
    & \qquad \dcomp (\id{A \tensor B} + \finmap{A} + \finmap{B} + 1) \dcomp (\id{A \tensor B} + 1 + \cocopy_{1}) \dcomp (\id{A \tensor B} + \cocopy_{1}) \\
    =& \ ((((\cp_{A} \tensor \id{X}) + \initmap{A \tensor X} + \initmap{A \tensor X}) \dcomp (\id{A} \tensor f) \dcomp (\id{A \tensor B} + \finmap{A})) \\
    &\qquad + ((\initmap{A \tensor X} + \initmap{1} + \id{X}) \dcomp f \dcomp (\finmap{B} + 1) \dcomp \cocopy_{1})) \dcomp (\id{A \tensor B} + \cocopy_{1})\\
    =& \ ((((\cp_{A} \tensor \id{X}) + \initmap{A \tensor X} + \initmap{A \tensor X}) \dcomp (\id{A} \tensor f) \dcomp (\id{A \tensor B} + \finmap{A})) + \finmap{X}) \\
    & \qquad \dcomp (\id{A \tensor B} + \cocopy_{1}),
  \end{align*}
  while the right side becomes:
  \begin{align*}
    &\ = ((\cp_{A+1} \dcomp (\id{A+1} \tensor \monadunit_{A+1})) \tensor \id{X}) \dcomp (\id{A+1} \tensor f) \dcomp \laxator_{A,B}\\
    &\ = ((\cp_{A} \tensor \id{X}) + \initmap{A \tensor X} + \initmap{A \tensor X} + \initmap{A \tensor X} + \id{X} + \initmap{1}) \dcomp ((\id{A} \tensor f) + f)\\
    & \qquad \dcomp (\id{A \tensor B} + \finmap{A} + \finmap{B} + 1) \dcomp (\id{A \tensor B} + 1 + \cocopy_{1}) \dcomp (\id{A \tensor B} + \cocopy_{1})\\
    &\ = ((((\cp_{A} \tensor \id{X}) + \initmap{A \tensor X} + \initmap{A \tensor X}) \dcomp (\id{A} \tensor f) \dcomp (\id{A \tensor B} + \finmap{A})) \\
    & \qquad + ((\initmap{A \tensor X} + \id{X} + \initmap{1}) \dcomp f \dcomp (\finmap{B} + 1) \dcomp \cocopy_{1})) \dcomp (\id{A \tensor B} + \cocopy_{1})\\
    & \ = ((((\cp_{A} \tensor \id{X}) + \initmap{A \tensor X} + \initmap{A \tensor X}) \dcomp (\id{A} \tensor f) \dcomp (\id{A \tensor B} + \finmap{A})) + \finmap{X}) \\
    & \qquad \dcomp (\id{A \tensor B} + \cocopy_{1}).
  \end{align*}
  The last equality in both the derivations holds because both the maps \((\initmap{A \tensor X} + \initmap{1} + \id{X}) \dcomp f \dcomp (\finmap{B} + 1) \dcomp \cocopy_{1}\) and \((\initmap{A \tensor X} + \id{X} + \initmap{1}) \dcomp f \dcomp (\finmap{B} + 1) \dcomp \cocopy_{1}\) have type \(X \to 1\), which means that they both coincide with the map \(\finmap{X}\) to the terminal object \(1\).
  This proves that the original equality holds.
\end{proof}

\begin{proof}[\textbf{Proof of \Cref{prop:pearlandjeffrey}}]
  Pearl's update on \(q\) is defined to be \(\bayesinv{(c \dcomp q)}{\sigma}\).
  This coincides with \(\bayesinv{q}{\sigma \dcomp c} \dcomp \bayesinv{c}{\sigma}\) by \Cref{prop:composition-bayes-inversions}.
  By applying the definition of Bayes inversion (\Cref{def:bayes-inversion}), \Cref{prop:QTconditionals-give-nice-marginals} and the partial Frobenius axioms (\Cref{diagram-partial-frobenius}), we show that Pearl's updated prior is Jeffrey's updated prior by a string diagrammatic reasoning similar to that of the proof of \Cref{th:bayes}.
  \pearlVSjeffreyupdatesFig{}
  Equality \((\ast)\) holds because \(y\) is deterministic.
\end{proof}

\begin{proof}[\textbf{Proof of \Cref{prop:unique-conditionals-collapse}}]
  If conditionals are unique, in particular, the conditional of \(\cp\) is unique.
  Both \(\id{} \tensor \discard\) and \(\discard \tensor \id{}\) are conditionals of \(\cp\) and they must coincide:
  \begin{equation}\label{eq:projections-coincide}
    \projectionscoincideFig{}
  \end{equation}
  Now, let \(f,g \colon X \to Y\) be two morphisms with the same probability of success:
  \begin{equation}\label{eq:same-probability-failure}
    \sameprobabilityfailureFig{}\ .
  \end{equation}
  We show that they must coincide.
  \uniqueconditionalscollapseproofFig{}
  Equations (\ref{eq:collapse-normal1}), (\ref{eq:collapse-normal2}), (\ref{eq:collapse-normal3}) and (\ref{eq:collapse-normal4}) follow from normalisation (\Cref{def:normalisation}), Equation (\ref{eq:collapse-proj1}) follows from \Cref{eq:projections-coincide}, and Equations (\ref{eq:collapse-hp1}), (\ref{eq:collapse-hp2}) and (\ref{eq:collapse-hp3}) follow from \Cref{eq:same-probability-failure}.
\end{proof}

\section{Proofs from \Cref{sec:observations-conditionals}}\label{sec:app-proofs-observations}

\begin{proof}[\textbf{Proof of \Cref{prop:embedding-cons-to-par}}]
  We define an identity-on-objects functor \(\fun{F} \colon \consproc{\cat{C}} \to \parproc{\cat{C}}\).
  Every morphism \(f\) in \(\consproc{\cat{C}}\) that comes from a morphism in \(\cat{C}\) is left unchanged: \(\fun{F}(f) \defn f\).
  For every deterministic state \(y \colon 1 \to Y\) in \(\cat{C}\), the image of its corresponding costate in \(\parproc{\cat{C}}\) is defined in \Cref{fig:map-deterministic-evidence}, using the comparator structure.
  The functor \(\fun{F}\) is defined freely on composite morphisms and it is well-defined because it preserves the axiom of \(\consproc{\cat{C}}\).
  This follows by the Frobenius axioms in \(\parproc{\cat{C}}\) (\Cref{diagram-partial-frobenius}) and the fact that \(y \colon 1 \to Y\) is deterministic:
  \consprocToParprocWelldefinedFig{}
  Moreover, if we assume that $y$ and $z$ are total, the following holds.
  \consprocToParprocFaithfulFig{}
  Equalities (\ref{eq:faithful-embed-frob1}) and (\ref{eq:faithful-embed-frob2}) follow from the Frobenius axioms (\Cref{diagram-partial-frobenius}), equalities (\ref{eq:faithful-embed-det1}) and (\ref{eq:faithful-embed-det2}) follow from determinism of \(y\) and \(z\), while equality (\ref{eq:faithful-embed-tot1}) follows from totality of \(y\) and \(z\).
  Equality (\ref{eq:faithful-embed-hp1}) is implied by the assumption \(\fun{F}(\obsv{y}) = \fun{F}(\obsv{z})\).
\end{proof}

\begin{proof}[\textbf{Proof of \Cref{prop:normalisation-of-programs}}]
  We start by observing that any \(g\) in \(\cat{C}\) that satisfies \Cref{eq:normal-form-programs} must be the normalisation of \(f\):
  \normalisationfromnormalformFig{}
  Now, we show that we can compute the normal form of \(f\) inductively.
  For the base cases, we note that all the generators are already in normal form, with either \(Y = \monoidalunit\) or \(Z = \monoidalunit\).

  For the inductive step, there are two possibilities.
  First, suppose that \(f\) is a composition of two \constrainedProcesses{}: \(f = f_1 \dcomp f_2\).
  By induction hypothesis, we can compute normal forms of \(f_1\) and \(f_2\).
  We compute the normal form of \(f\) by combining the normal forms of \(f_1\) and \(f_2\), computing conditionals of total morphisms and using the axiom in \Cref{fig:axiom-exact-conditioning-cat}.
  \normalformcompositionprobabilisticprogramsProofFig{}
  The second possibility is that \(f\) is a monoidal product of two \constrainedProcesses{}: \(f = f_1 \tensor f_2\).
  By induction hypothesis, we can compute normal forms of \(f_1\) and \(f_2\).
  The normal form of \(f\) is the monoidal product of the normal forms of \(f_1\) and \(f_2\).
  \normalformtensorprobabilisticprogramsProofFig{}
\end{proof}

\section{Some decision problems}\label{sec:app-decision-problems}
\subsection{Death in Damascus}\label{sec:death-damascus}
In a completely deterministic world, Death collects people on a designated place on a designated day. If the chosen people is not there to confront Death, they survive (which represents a great utility, say, \(1000 \utility\)).

The legend says that a merchant found Death in Damascus, and Death promised to come for him in the next day. The merchant thought of fleeing to Aleppo, trying to escape death; but that came with a cost (a small negative utility, say, \(-1 \utility\)). However, Death is a perfect predictor, so the merchant found Death in Aleppo. Should the merchant have fled to Aleppo?

\begin{figure}[h!]
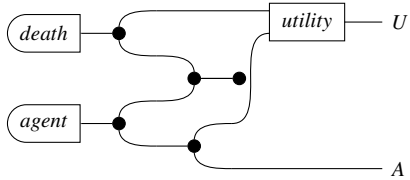

  \[\deathdamascusFig{}\]
    \caption{Model for Death in Damascus.}\label{diagram-damascus-model}
\end{figure}

Consider the model in \Cref{diagram-damascus-model}.
Evidential Decision Theory prescribes just waiting for Death in Damascus. In this model, if Death is really omniscient, it will be impossible to avoid it. It only makes sense to avoid the small negative utility of a last trip to Aleppo, accepting \(0 \utility\).

\begin{center}
  \begin{tabular}{| c | c  c |}
    \hline
    & \(D = \ket{A}\) & \(D = \ket{D}\) \\
    \hline
    \(A = \ket{A}\) & \(0 \utility\) & \(999 \utility\) \\
    \(A = \ket{D}\) & \(1000 \utility\) & \(-1 \utility\)\\
    \hline
  \end{tabular}
\end{center}

\subsection{Cheating Death in Damascus with a random oracle}\label{sec:cheating-death}
The reader may observe this problem could have been also modelled in a way similar to Newcomb's.
The only difference is that, in this case, the predictor is always \emph{adversarial}.
How could the agent cheat against such a predictor?
A possible answer is to allow the agent to use a true random oracle: if it were to decide whether to flee to Damascus or Aleppo based on a random oracle that even Death cannot predict, it would still have a chance of cheating Death.

In this second formulation of the problem (\Cref{diagram-damascus-coin-model}), the agent can use a fair coin that Death cannot predict.
We can even allow Death to choose the same strategy and toss a coin as well.
Should the agent try to cheat Death and choose to toss a coin?

\begin{figure}[h!]
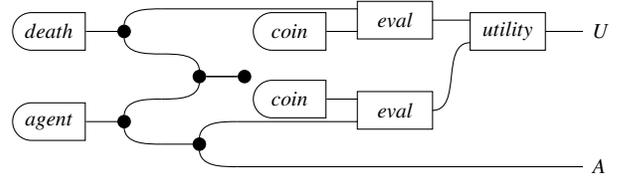

  \[\cheatingdeathdamascusFig{}\]
    \caption{Model for Cheating Death in Damascus with a random coin.}\label{diagram-damascus-coin-model}
\end{figure}

Evidential Decision Theory now prescribes that the agent should use the coin to try to cheat Death. This is no longer a lost cause: the expected utility is now $445.5 \utility$, see \Cref{subsec:codeDamascus}.

\section{Implementation}\label{app:implementation}

\lstset{basicstyle=\ttfamily,breaklines=true}

\subsection{Newcomb's Problem}\label{subsec:codeNewcomb}
The following is the model for Newcomb's Problem.
An agent will take action \texttt{a} with an uninformative prior. A predictor will try to predict it with \texttt{p}, again using an uninformative prior. We observe that the prediction is correct. Which is the action \texttt{x} that we would like to observe we have chosen?
\begin{lstlisting}[linerange={25-50}]
newcomb :: Action -> Distribution Value
newcomb x = do
    a <- act
    p <- predict
    observe (a == p)
    observe (a == x)
    return $ outcome a p

act :: Distribution Action
act = Distribution \case
    OneBox -> 1 / 2
    TwoBox -> 1 / 2

predict :: Distribution Action
predict = Distribution \case
    OneBox -> 1 / 2
    TwoBox -> 1 / 2

outcome :: Action -> Action -> Value
outcome OneBox OneBox = Hundred
outcome OneBox TwoBox = Zero
outcome TwoBox OneBox = HundredOne
outcome TwoBox TwoBox = One


\end{lstlisting}
Our program will evaluate \lstinline{argmax newcomb} to the answer \lstinline{OneBox}.

\subsection{Death in Damascus}\label{subsec:codeDamascus}
The following is the model for the ``Death in Damascus'' problem. We sample a \lstinline{merchant} from the population of the world, and this information is also known by \lstinline{death}, who uses it to decide which city to go to. The merchant throws a \lstinline{coin} and chooses whether to flee or to stay following some strategy. Which is the strategy \texttt{f} that we would like to observe the merchant to have chosen?
\begin{lstlisting}[linerange={20-62}]
deathInDamascus :: Strategy -> Distribution Outcome
deathInDamascus f = do
    merchant <- world
    cityDeath <- death merchant
    coin <- throwCoin
    cityMerchant <- return (flee (strategy merchant) coin)
    observe (f == strategy merchant)
    return $ outcome cityMerchant cityDeath
    
world :: Distribution Merchant
world = Distribution \case
    FleeingMerchant -> 1 / 3
    StayingMerchant -> 1 / 3
    RandomMerchant  -> 1 / 3

strategy :: Merchant -> Strategy
strategy FleeingMerchant = Fleeing
strategy StayingMerchant = Staying
strategy RandomMerchant  = Random

throwCoin :: Distribution Coin
throwCoin = Distribution \case
    Heads -> 1 / 2
    Tails -> 1 / 2

flee :: Strategy -> Coin -> City
flee Fleeing _ = Aleppo
flee Staying _ = Damascus
flee Random Heads = Aleppo
flee Random Tails = Damascus 

death :: Merchant -> Distribution City
death FleeingMerchant = return Aleppo
death StayingMerchant = return Damascus
death RandomMerchant = Distribution \case
    Damascus -> 1 / 2
    Aleppo   -> 1 / 2

outcome :: City -> City -> Outcome
outcome Aleppo Aleppo     = MerchantTravelsAndMeetsDeath
outcome Aleppo Damascus   = MerchantTravelsAndEscapes
outcome Damascus Aleppo   = MerchantEscapes
outcome Damascus Damascus = MerchantMeetsDeath
\end{lstlisting}

\subsection{Partial Markov Category of Subdistributions}
The following is the library for Evidential reasoning using the partial Markov category of subdistributions. The subdistribution monad is better modelled here as a relative monad \cite{altenkirch_chapman_uustalu} from \lstinline{Finitary} types to arbitrary types. We employ rebindable syntax in order to be able to use do-notation \cite{haskell_2007} for the Kleisli category of this relative monad.
\begin{lstlisting}
-- A trial implementation of a probabilistic programming 
-- setup with the primitives of a Partial Markov Category.

{-# LANGUAGE GADTs #-}
{-# LANGUAGE RebindableSyntax #-}
{-# LANGUAGE DataKinds #-}
{-# LANGUAGE DeriveAnyClass #-}
{-# LANGUAGE DeriveGeneric #-}
{-# LANGUAGE DerivingStrategies #-}
{-# LANGUAGE LambdaCase #-}
{-# LANGUAGE BlockArguments #-}
{-# LANGUAGE InstanceSigs #-}


module Bayes where

import Prelude hiding ((>>=), (>>), return)
import Data.Finitary
import Data.Ord
import Data.List ( maximumBy )
import GHC.Generics (Generic)



-- Finitary subdistribution monad.
data Distribution a where
  Distribution :: (Finitary a) => (a -> Rational) -> Distribution a

-- Properties of a distribution.
total :: Distribution a -> Rational
total (Distribution d) = sum [ d x | x <- inhabitants ]

weight :: Distribution a -> a -> Rational
weight (Distribution f) a = f a

-- Normalization here is just a pretty printing thing.
normalize :: Distribution a -> Distribution a
normalize (Distribution f) = Distribution $ \a ->
  f a / (total (Distribution f))
instance (Finitary a, Show a) => Show (Distribution a) where
  show :: (Finitary a, Show a) => Distribution a -> String
  show d =
    let (Distribution f) = normalize d in
      unlines [ show (f a, a) | a <- inhabitants ]


-- Rebindable do notation.
(>>=) :: (Finitary a , Finitary b) => 
  Distribution a -> (a -> Distribution b) -> Distribution b
(>>=) (Distribution d) f = Distribution $ \b ->
  sum $ [ (d a) * (weight (f a) b) | a <- inhabitants]

(>>) :: (Finitary a, Finitary b) => 
  Distribution a -> Distribution b -> Distribution b
(>>) d (Distribution f) = Distribution $ \b -> (total d) * (f b)

fail :: Distribution a -> String
fail = undefined

ifThenElse :: Bool -> a -> a -> a
ifThenElse True x _ = x
ifThenElse False _ y = y 

-- Distribution combinators.
return :: (Finitary a) => a -> Distribution a
return x = Distribution (\y ->
  case (x == y) of
    True -> 1
    False -> 0)

absurd :: (Finitary a) => Distribution a
absurd = Distribution (\a -> 0)

observe :: Bool -> Distribution ()
observe True = return ()
observe False = absurd


class (Finitary a) => Valuable a where
  value :: a -> Rational

expectedValue :: (Valuable a) => Distribution a -> Rational
expectedValue u = 
  let (Distribution d) = normalize u in 
    sum $ fmap (\x -> value x * d x) inhabitants

argmax :: (Finitary a, Valuable b) => (a -> Distribution b) -> a
argmax f = maximumBy (comparing (expectedValue . f)) inhabitants
\end{lstlisting}
 
\end{document}